\documentclass[conf]{new-aiaa}
\usepackage[utf8]{inputenc}
\usepackage{textcomp}

\usepackage{graphicx}
\usepackage{amsmath}
\usepackage[version=4]{mhchem}
\usepackage{siunitx}
\usepackage{longtable,tabularx}
\usepackage{verbatim}
\usepackage{bm}
\usepackage{amsthm}
\usepackage{mathrsfs}
\usepackage{bbm}
\usepackage{amssymb}
\usepackage{footnpag}
\usepackage{threeparttable}
\usepackage{caption}
\usepackage{subcaption}
\usepackage{booktabs}

\newcommand{\R}{\mathbbm{R}}

\newcommand{\sN}{\mathcal{N}}

\usepackage{pifont}
\newcommand{\cmark}{\ding{51}}%
\newcommand{\xmark}{\ding{55}}%

\newtheoremstyle{italicblock}
  {12pt}
  {12pt}
  {\itshape}
  {}
  {\bfseries}
  {}
  {0.5em}
  {}

\theoremstyle{italicblock}

\numberwithin{problem}{section}

\newtheorem{proposition}{Proposition}
\numberwithin{proposition}{section}

\newtheorem{theorem}{Theorem}
\numberwithin{theorem}{section}

\newtheorem{definition}{Definition}
\numberwithin{definition}{section}

\newtheorem{assumption}{Assumption}
\numberwithin{assumption}{section}

\numberwithin{approximation}{section}

\newtheorem{remark}{Remark}
\numberwithin{remark}{section}

\title{Real-Time Stochastic Terrain Mapping and Processing for Autonomous Safe Landing}

\author{Kento Tomita\footnote{Ph.D. in Aerospace Engineering; currently Research Scientist, Mitsubishi Electric Research Laboratories, Cambridge, Massachusetts 02139; tomita@merl.com}}
\affil{Georgia Institute of Technology, Atlanta, GA, 30332}
\author{Koki Ho\footnote{Associate Professor, Daniel Guggenheim School of Aerospace Engineering, 620 Cherry St. NW, Atlanta, GA, AIAA Associate Fellow.}}
\affil{Georgia Institute of Technology, Atlanta, GA, 30332}

\begin{document}

\maketitle

\begin{abstract}
Onboard terrain sensing and mapping for safe planetary landings often suffer from missed hazardous features, e.g., small rocks, due to the large observational range and the limited resolution of the obtained terrain data. 
To this end, this paper develops a novel real-time stochastic terrain mapping algorithm that accounts for topographic uncertainty between the sampled points, or the uncertainty due to the sparse 3D terrain measurements. We introduce a Gaussian digital elevation map that is efficiently constructed using the combination of Delauney triangulation and local Gaussian process regression. The geometric investigation of the lander-terrain interaction is exploited to efficiently evaluate the marginally conservative local slope and roughness while avoiding the costly computation of the local plane. The conservativeness is proved in the paper. The developed real-time uncertainty quantification pipeline enables stochastic landing safety evaluation under challenging operational conditions, such as a large observational range or limited sensor capability, which is a critical stepping stone for the development of predictive guidance algorithms for safe autonomous planetary landing. Detailed reviews on background and related works are also presented. 
\end{abstract}

\section{Introduction}
\lettrine{R}{eal-time} uncertainty quantification (UQ) of landing safety is a critical component for the successful deployment of planetary landers. 
Consider spacecraft approaching a planetary surface, executing powered descent, and seeking safe landing sites through terrain sensing. 
Terrain sensing, however, is compromised by noisy measurements and the inherent uncertainties in the lander's state, leading to significant ambiguity in the estimated terrain topography and landing safety. 
The guidance system of the lander is tasked with optimizing the landing site and associated trajectory based on a stochastic safety map generated by the hazard detection (HD) module. 
Given precise UQ of the terrain and safety, the guidance system can either pinpoint the safest landing site or navigate the lander toward a safer region for additional onboard terrain sensing.

However, achieving real-time UQ of the terrain and safety presents formidable challenges. 
The primary source of uncertainty arises from the sparse 3D measurements of the terrain; a situation where missing information cannot merely be approximated by simple additive noise models, which are typically adequate for other uncertainties like measurement noise or the uncertainties in the lander's position and attitude.
For accurate hazard detection, such as identifying small rocks, it is imperative to achieve a resolution in the measured 3D terrain data that is at least half the maximum tolerable rock height, presenting a significant technical challenge. 
Active terrain sensors, such as LiDAR, provide direct measurements of 3D locations on the terrain surface, but these measurements are inherently sparse, especially in planetary landing scenarios where terrain sensing ideally begins at altitudes of 1km or 2km. 
Conversely, passive sensors like optical cameras offer higher resolution, but reconstructing a 3D scene from images often fails to achieve the necessary resolution due to the terrain's poor textures.

This paper presents a novel real-time stochastic hazard detection algorithm capable of handling topographic uncertainty arising from sparse terrain measurement. First, we develop a novel real-time Gaussian DEM construction algorithm that utilizes Delaunay triangulation and a localized Gaussian Random Field (GRF) model for efficient evaluation.
Next, we introduce an innovative real-time stochastic hazard detection algorithm that converts Gaussian DEMs into stochastic safety maps. A fast and provably conservative HD algorithm is developed for standard DEM input, which is then extended to accommodate Gaussian DEM input. 
Finally, the effectiveness of these proposed methodologies is demonstrated through a combination of real and synthetic terrain data. The proposed pipeline is shown in Fig.~\ref{fig:pipeline}

\begin{figure}
    \centering
    \includegraphics[width=0.9\linewidth]{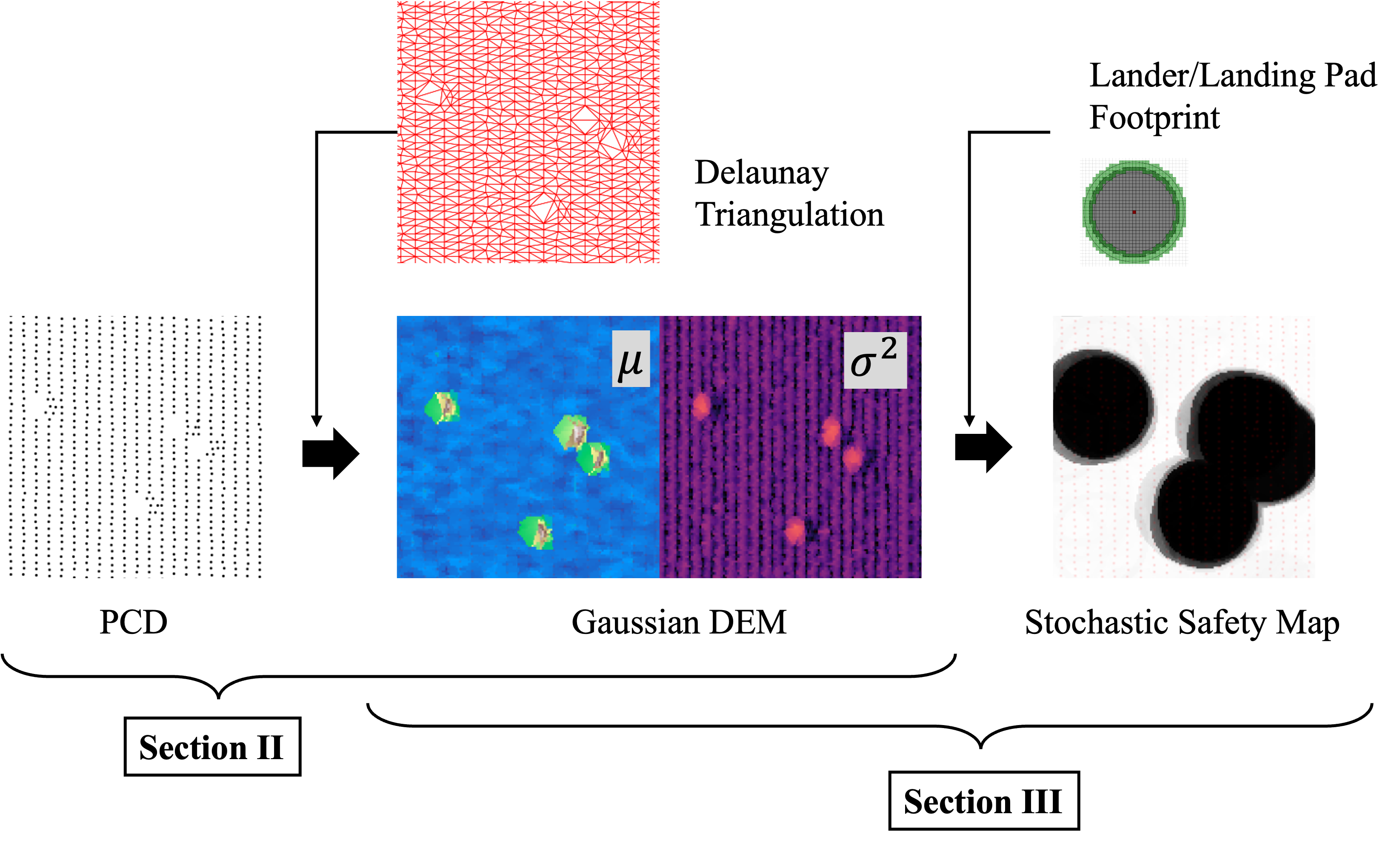}
    \caption{Proposed pipeline for real-time stochastic terrain mapping and processing for safe landing.}
    \label{fig:pipeline}
\end{figure}

\subsection{Background}
Historically, the deployment of planetary landers has been constrained by the absence of onboard hazard detection and avoidance (HDA) capabilities. This limitation necessitated a reliance on mechanical hazard tolerance~\cite{golombek1997SelectionMarsPathfinder} and the selection of landing sites in terrains deemed exceptionally safe, where 95-99\% of the area within the landing ellipse was required to be hazard-free~\cite{spencer2009PhoenixLandingSite}. However, the reliance on statistical estimations of landing safety—predominantly based on limited reconnaissance data from remote sensing—has proven to be inadequate. A notable instance involves NASA's Viking 1 Mars lander in 1976. Despite thorough planning and examination of the landing site~\cite{masursky1976VikingLandingSites}, the lander was unexpectedly close to a significant obstacle, a large boulder referred to as \textit{Big Joe}, posing a potential risk to the mission's success (Fig. \ref{fig:big-joe}~\cite{jplphoto}).
The necessity for autonomous HDA was underscored during the Apollo program, which was the first demonstration of real-time HDA via pilot-in-the-loop guidance and navigation.
The Apollo 15 mission in 1971, in particular, highlighted the challenges associated with manual HDA~\cite{brady2009HazardDetectionMethods}; even under near-ideal lighting conditions, navigating and ensuring a safe landing was demanding, resulting in the dramatic tilt of the spacecraft that was only $1^{\circ}$ away from the tolerance (Fig. \ref{fig:apollo-15}~\cite{apollo15}).

\begin{figure}
     \centering
     \begin{subfigure}[b]{0.4\textwidth}
         \centering
         \includegraphics[width=\textwidth]{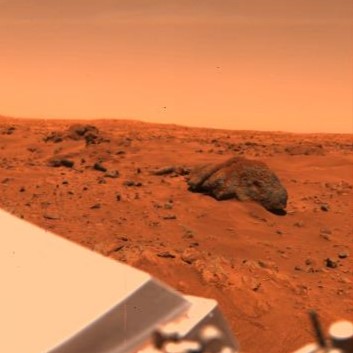}
         \caption{}
         \label{fig:big-joe}
     \end{subfigure}
     \begin{subfigure}[b]{0.4\textwidth}
         \centering
         \includegraphics[width=\textwidth]{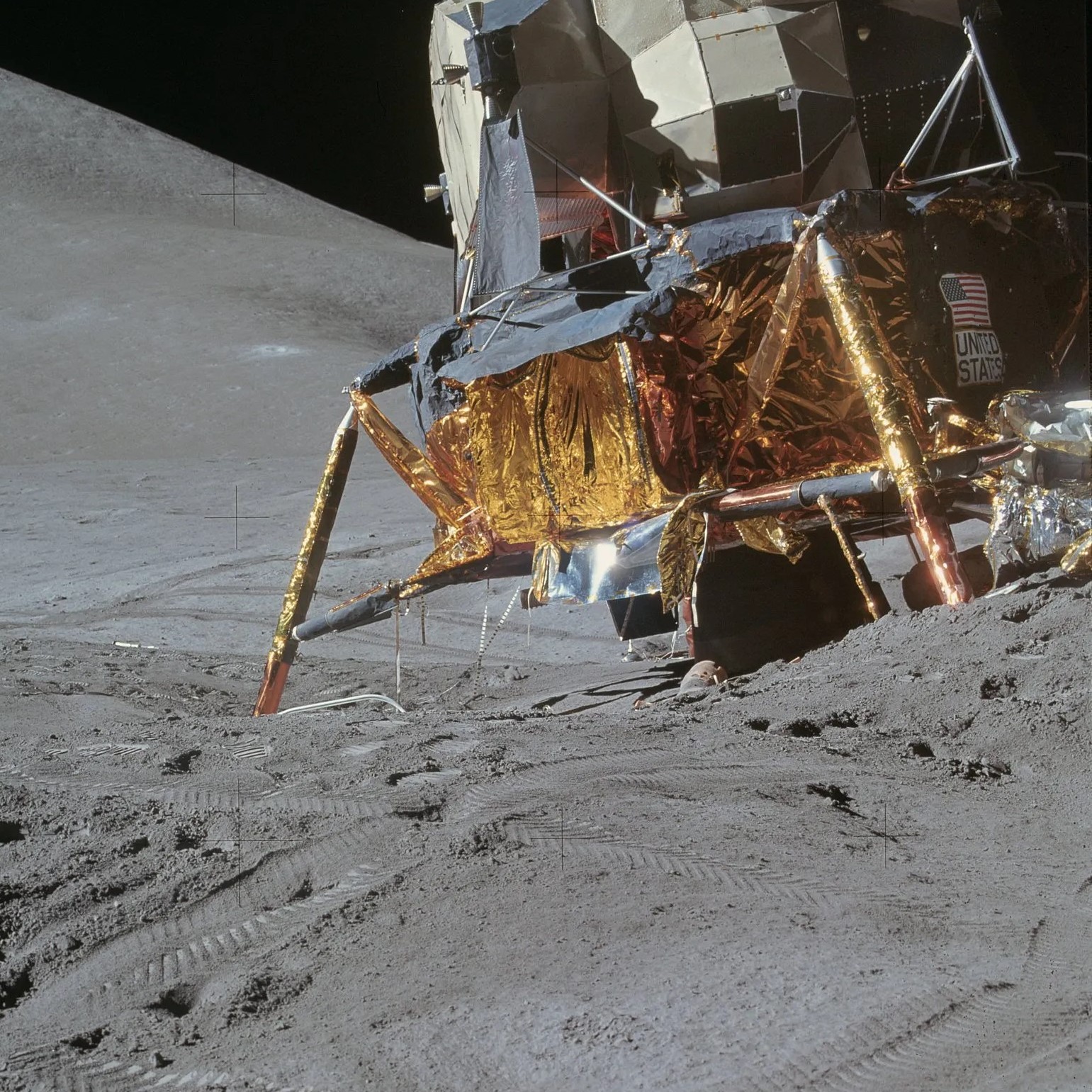}
         \caption{}
         \label{fig:apollo-15}
     \end{subfigure}
        \caption{(a) A boulder, ``Big Joe,'' next to the Viking 1 Mars lander~\cite{jplphoto}. (b) The Apollo 15 lunar lander straddling next to a small crater~\cite{apollo15}.}        \label{fig:three graphs}
\end{figure}

Although autonomous spacecraft landing has been an active research topic since the early days of space exploration~\cite{schappell1976ExperimentalSimulationStudy}, the inception of the Space Exploration Initiative (SEI) in 1989 drove further investigation, with NASA identifying two pivotal technological capabilities essential for autonomous landing: Autonomous Precision Landing (APL) and Autonomous Hazard Detection and Avoidance (AHDA)~\cite{tchorykjr1991PassiveActiveSensors}. These concepts have evolved into the current framework of NASA's Precision Landing and Hazard Avoidance (PL\&HA) strategy. Central to APL is the need for precise navigation relative to a preselected site, a concept now widely recognized as terrain relative navigation (TRN), and the detailed reconnaissance data providing the preferable landing site. Meanwhile, AHDA hinges on the capabilities of hazard detection sensors and algorithms tailored to assess and mitigate landing risks on the spot.

The balance between APL and AHDA technologies required is influenced by the amount of prior knowledge available about the landing surface, and sensor systems and algorithms available. In scenarios where the surface is well-characterized, APL might be sufficient, making AHDA less critical. Conversely, with limited or uncertain information about the landing terrain, AHDA becomes indispensable to ensure mission success~\cite{john1990AutonomousLandingMars}. Early research highlighted the trade-offs between active and passive sensors for APL and AHDA, advocating for a hybrid approach that combines the strengths of both sensor types~\cite{tchorykjr1991PassiveActiveSensors, juday1989HybridVisionAutomated}. Active sensors, such as LiDARs and radars, provide detailed three-dimensional topographic data but at the cost of increased size, weight, and power (SWaP) requirements, and may also be constrained by their operational range and resolution. Passive sensors, e.g., optical cameras, offer a low-resource alternative but lack direct topographic measurement capabilities and are significantly affected by lighting conditions.

The landmark achievements of CNSA's Chang'e-3 lunar lander in 2013, with its successful AHDA, and NASA's Mars 2020 mission, demonstrating APL for a safe landing on Mars, exemplify the practical realization of these technologies~\cite{li2016GuidanceSummaryAssessment, johnson2023ImplementationMapRelative}. Chang'e-3 implemented a two-stage AHDA process~\cite{li2016GuidanceSummaryAssessment}, initially employing coarse image-based hazard detection~\cite{bajracharya2002SingleImageBased, cheng2001PassiveImagingBased} at higher altitudes, followed by precise LiDAR-based assessments~\cite{johnson2002LidarBasedHazardAvoidance, johnson2008AnalysisOnBoardHazard} during a hovering phase at approximately 100 meters. This methodology was subsequently applied with success in the Tianwen-1 Mars landing in 2021~\cite{huang2023Tianwen1EntryDescent}. Conversely, Mars 2020 leveraged APL~\cite{johnson2023ImplementationMapRelative}, utilizing a detailed reference map created from reconnaissance data for TRN, and successfully avoided pre-identified hazards, achieving an overall positional accuracy within 5 meters of the targeted location~\cite{nelessen2019Mars2020Entry, cheng2021MakingOnboardReference}. Recent missions, including ISRO's Chandrayaan-3, JAXA's SLIM, and NASA's CLPS, have integrated APL and AHDA techniques, utilizing reconnaissance data and TRN to navigate spacecraft to known, relatively safer regions while performing AHDA to circumvent unforeseen hazards before touchdown~\cite{suresh2024ApproachChandrayaan3Lander, karanam2023ContextualCharacterisationStudy, getchius2022HazardDetectionAvoidance, slimwebpage}.

As the feasibility of autonomous spacecraft landing is increasingly demonstrated, the demand for even more advanced capabilities grows, with the aim of achieving anytime and anywhere global and safe landing capabilities across the solar system—a milestone articulated by NASA's PL\&HA program~\cite{carson2022NASADevelopmentStrategy}. A notable challenge in this endeavor is the need for real-time mapping to support active TRN and HDA under any lighting condition during descent~\cite{NASA-LAND}. Addressing this requires the development of advanced sensor suites that balance performance with manageable SWaP demands, alongside efficient processing algorithms. This paper introduces a novel approach to stochastic terrain mapping and real-time algorithms that enable efficient terrain analysis and landing safety evaluation with limited data. This method not only potentially reduces the requirements for HD sensors, e.g., range and resolution, but also relaxes their effective operational conditions. Moreover, it offers innovative tools for uncertainty quantification in landing safety assessment, paving the way for new stochastic guidance algorithms that can accommodate landing site uncertainty during the final stages of AHDA.

\subsection{Related Works and Preliminaries}

Hazard detection (HD) algorithms, often referred to as terrain mapping or landing site selection algorithms, can be broadly categorized into vision-based methods~\cite{huertas2006PassiveImagingBased, huertas2007RealtimeHazardDetection} and LiDAR-based methods~\cite{johnson2002LidarBasedHazardAvoidance}. Vision-based methods can be further divided into two groups: those that directly generate the landing safety map from visual input and those that first reconstruct the terrain surface that is then processed to obtain the safety map. The first group includes techniques like texture analysis~\cite{cheng2001PassiveImagingBased} and shape-from-shading~\cite{zhang1999ShapefromshadingSurvey}, which detect hazardous roughness features such as rocks, and homography transformation, which evaluates local slope~\cite{cheng2001PassiveImagingBased}. 
Many machine learning-based approaches have also been explored to directly transform the terrain images into safety maps~\cite{lunghi2016MultilayerPerceptronHazard, huang2019VisionBasedHazardDetection, driver2023DeepMonocularHazard, ghilardi2023ImageBasedLunarHazard, petrakis2023LunarGroundSegmentation, suwinski2024ImageBasedLanding} and database for training machine learning (ML) algorithms has been published as well~\cite{yu2017DatabaseConstructionVision}.
The second group of vision-based methods reconstructs the 3D terrain using feature-based approaches such as optical flow~\cite{hoff1990PlanetaryTerminalDescent} or stereo vision~\cite{woicke2016StereovisionHazarddetectionAlgorithm}. Recent advancements in computational power and computer vision algorithms have led to the application of structure-from-motion techniques~\cite{olson2001MultiresolutionMappingUsing, schoppmann2021MultiResolutionElevationMapping, getchius2022HazardDetectionAvoidance} and machine learning-based dense feature matching~\cite{posada2024DenseFeatureMatching} to reconstruct terrain surfaces at finer resolutions. Once the 3D terrain surface representation is obtained, LiDAR-based HD algorithms can be applied to generate safety maps.

With the 3D terrain information, LiDAR-based HD algorithms evaluate landing safety based on the lander's specification for the tolerable slope and roughness at touchdown, which are two major factors for topographic landing safety for planetary landers. Both slope and roughness are defined by the geometric relationship between the local terrain and the \textit{landing surface}, which is the plane formed by the landing legs in contact with the terrain (Fig. \ref{fig:slope_rghns_def}). The slope is the tilting angle of the surface normal vector of the landing surface with respect to the local gravitational vector, while the roughness is the distance between the landing surface and surface obstacles. Planetary landers are designed to avoid tipping over under certain slope conditions and specific terminal horizontal velocity constraints. The configuration between the lander's bottom instruments and the landing legs defines the tolerable size of the obstacle underneath the lander. Given the lander's slope and roughness tolerance, the HD algorithm evaluates whether each landing location offers a tolerable slope and roughness.

\begin{figure}
    \centering
    \includegraphics[width=0.7\linewidth]{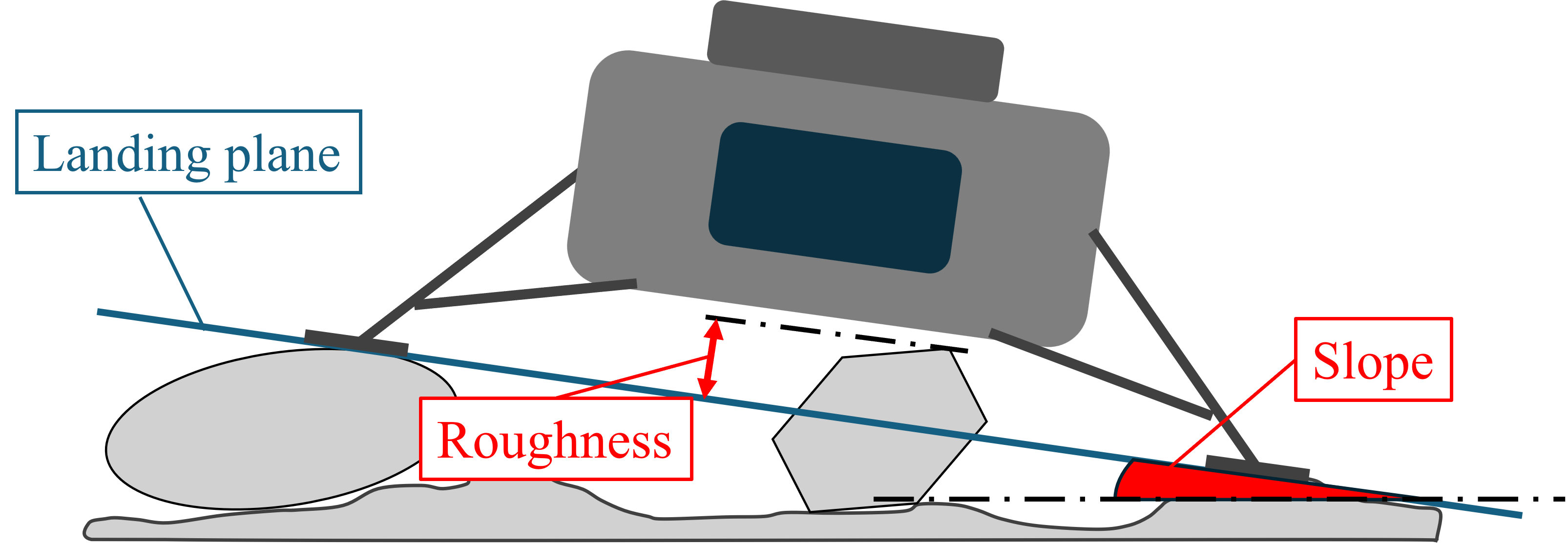}
    \caption{Definition of landing hazards: slope and roughness.}
    \label{fig:slope_rghns_def}
\end{figure}

LiDAR-based HD algorithms can be categorized based on the methods for computing the \textit{reference plane}, which is used to calculate slope and roughness (Fig. \ref{fig:reference-plane-comparison}). For precise safety evaluation, the reference plane must be the landing surface. However, the exact computation of the landing surface involves evaluating the lander-terrain interaction for all orientations, which is computationally impractical. As a result, many existing works approximate the slope and roughness by using the estimated "base" terrain as a reference plane. The base terrain is the imaginary terrain without surface obstacles such as rocks, where the true terrain is considered the superposition of the base terrain and those surface obstacles.

The pioneering work of the LiDAR-based HD algorithm by Johnson et al.~\cite{johnson2002LidarBasedHazardAvoidance, johnson2008AnalysisOnBoardHazard, johnson2010AnalysisFlashLidar} proposed estimating the base terrain by applying the least median square (LMedsq) fitting, instead of the regular least square fitting~\cite{CUI2017326}. Later, Xiao et al.~\cite{xiao2015RobustPlaneFitting} identified that the LMedsq approach is vulnerable to a large number of outliers and developed a robust estimator for the reference plane by combining the M-estimator sample consensus (MSAC) with a modified iterative Kth ordered scale estimator (IKOSE). In 2019, Xiao et al. applied thin plane spline (TPS) interpolation for computing the base terrain more efficiently, avoiding repeated local plane fittings over different locations~\cite{xiao2019SafeMarsLanding}. Note that for a coarse digital elevation map whose resolution is as large as the lander diameter, we may skip the terrain fitting and take finite difference-based derivatives to estimate the local slope~\cite{LIU2019272}.

Those HD algorithms with the base terrain approximation do not provide accurate slope and roughness evaluation upon touchdown, as illustrated in Fig. \ref{fig:reference-plane-comparison}. To address this, in the Autonomous Landing and Hazard Avoidance Technology (ALHAT) program, NASA's Jet Propulsion Laboratory (JPL) proposed a more exact HD algorithm that evaluates deterministic slope and probabilistic roughness safety over the terrain using the landing plane as a reference plane~\cite{ivanov2013ProbabilisticHazardDetection}. The improved performance of the proposed algorithm was demonstrated in field tests~\cite{luna2017EvaluationSimpleSafe}, but the algorithm has proved computationally too expensive, especially for high-resolution inputs~\cite{johnson2022OPTIMIZATIONLIDARBASED}.

\begin{figure}
    \centering
    \includegraphics[width=1\linewidth]{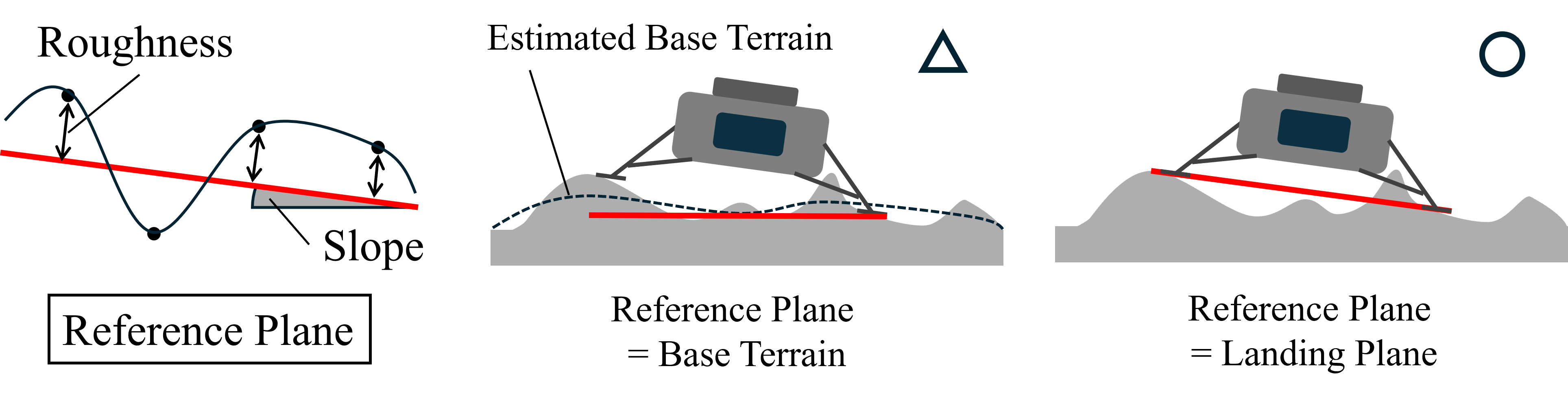}
    \caption{Comparison of reference planes for slope and roughness evaluation.}
    \label{fig:reference-plane-comparison}
\end{figure}

Recent advancements in deep learning (DL) have been applied to achieve real-time computation and improved performance. Tomita et al.~\cite{tomita2020HazardDetectionAlgorithm} applied deep semantic segmentation techniques to estimate exact hazard evaluation based on the landing plane, showing it outperforms the ALHAT algorithm~\cite{ivanov2013ProbabilisticHazardDetection} in terms of computational time and robustness to increased measurement noise. Later, Tomita et al.~\cite{tomita2022BayesianDeepLearning} applied Bayesian deep learning to enable the probabilistic assessment of the DL's output for improved reliability. Other LiDAR-based HD algorithms using DL approaches~\cite{moghe2020DeepLearningApproach, zha2021LandingSiteSelection} have shown their performance for simplified HD tasks, typically evaluating only slope and ignoring roughness safety, which is harder to assess precisely.

The most typical data structure for LiDAR-based HD algorithms has been the \textit{digital elevation map (DEM)} with a fixed resolution. DEMs store height information as raster data, which allows for efficient access to 3D information compared to other 3D data types like triangular meshes. However, DEMs have a fixed resolution and limited expressiveness. Recently, variable-resolution data structures have gained attention for their flexibility and potential efficiency improvements.

Mango et al. proposed a two-stage terrain mapping algorithm to reduce the computational cost of processing dense point cloud data (PCD). Initially, it performs coarse voxelization and measures the slope of the local normal by RANSAC-based plane fitting. Roughness is evaluated as the standard deviation of the PCD from the reference ground plane, also obtained by RANSAC-based plane fitting. For safe pixels, fine HD is performed using the same approach but in finer resolution~\cite{mango2020HazardDetectionLanding}. Schoppmann et al. presented an algorithm for multiresolution landing site detection using a multiresolution map generated by an onboard structure-from-motion program. The map is organized into a Laplacian pyramid, dividing terrain features into layers based on their frequency content, incorporating slope, roughness, and mapping quality~\cite{schoppmann2021MultiResolutionElevationMapping}. Setterfield et al. proposed a quadtree navigation and mapping system by dividing the quadtree based on measurement density~\cite{setterfield2021LiDARInertialBasedNavigation}. Marcus et al. proposed a filter-based terrain mapping algorithm using a quadtree, where each cell stores the estimated local plane obtained through maximum likelihood (ML) estimation. The DEM resolution is increased upon receiving new measurements that deviate from the estimated plane, allowing natural data fusion~\cite{marcus2024LandingSiteMapping}. 

Jung et al. proposed a unique method involving the fitting of constant elevation contour lines to a triangular Digital Terrain Map (DTM)~\cite{jung2020DigitalTerrainMap}. This technique is used to detect rocks and measure slopes. However, this approach encounters several challenges, primarily related to the interval selection between contour lines, which can lead to hazards being overlooked. To mitigate these challenges, numerous conditional checks are required, making the algorithm's decision tree quite complex.

Table \ref{tab:algorithms} presents the comparison of LiDAR-based HD algorithms. Most of the existing works are not based on exact landing hazards and their performance is often evaluated using qualitative results rather than rigorous metrics. Compared to existing works where algorithm design and analysis are based on exact landing hazards, the proposed algorithm achieves improved efficiency, stochastic evaluation, and dense evaluation that enables quantifying landing safety between sampled points.

\begin{table}[ht]
\caption{\label{tab:algorithms} Comparison of LiDAR-based Hazard Detection Algorithms}
\centering
\begin{tabular}{@{}lcccc@{}}
\toprule
{Algorithm} & 
{Real-Time} &
{Exact HD} & 
{Stochastic} &
{Dense Evaluation} \\ \midrule \hline
Johnson et al. ~\cite{johnson2002LidarBasedHazardAvoidance} & \cmark & \xmark & \xmark & \xmark\\
Xiao et al. (2015)~\cite{xiao2015RobustPlaneFitting} & - & \xmark & \xmark & \xmark \\
Xiao et al. (2019)~\cite{xiao2019SafeMarsLanding} & - & \xmark & \xmark & \xmark \\
Jung et al. ~\cite{jung2020DigitalTerrainMap} & - & \xmark & \xmark & \xmark \\
Moghe et al. ~\cite{moghe2020DeepLearningApproach} & - & \xmark & \xmark & \xmark \\
Zha et al. ~\cite{zha2021LandingSiteSelection} & - & \xmark & \xmark & \xmark \\
Mango et al. ~\cite{mango2020HazardDetectionLanding} & - & \xmark & \cmark & \xmark \\
Schoppmann et al. ~\cite{schoppmann2021MultiResolutionElevationMapping} & \cmark & \xmark & \cmark & \xmark \\
Setterfield et al. ~\cite{setterfield2021LiDARInertialBasedNavigation} & \cmark & \xmark & \cmark & \xmark \\
Marcus et al. ~\cite{marcus2024LandingSiteMapping} & \cmark & \xmark & \cmark & \xmark \\
Ivanov et al. ~\cite{ivanov2013ProbabilisticHazardDetection} & \xmark$^*$ & \cmark & \cmark & \xmark \\
Tomita et al. (2020)~\cite{tomita2020HazardDetectionAlgorithm} & \cmark & \cmark & \xmark & \xmark \\
Tomita et al. (2022)~\cite{tomita2022BayesianDeepLearning} & - & \cmark & \cmark & \xmark \\
This work & \cmark & \cmark & \cmark & \cmark \\
\bottomrule
\end{tabular}
\begin{tablenotes}
\item {\footnotesize $^*$ Based on results given in ~\cite{johnson2022OPTIMIZATIONLIDARBASED} and this paper.}
\end{tablenotes}
\end{table}

\section{Real-Time Gaussian DEM Construction}
This section covers the topographic data creation from the terrain sensing measurement. First, we review the conventional DEM construction algorithm. Then, the topographic uncertainty model based on a GRF is presented. Finally, we propose the real-time Gaussian DEM construction algorithm via local regression based on Delaunay triangulation.

\subsection{Conventional DEM Construction Algorithm}\label{sec:conventional-pcd2dem}
For HD algorithms, the most typical data structure for topographic data has been a DEM with a fixed resolution. DEMs store height information as raster data, which allows for efficient access to 3D information compared to other 3D data types like triangular meshes. This raster grid representation often enables fast processing time and is suitable for real-time HD applications.

Traditionally, the resolution of a DEM is determined by the average GSD of the PCD to make the most use of the samples while minimizing the pixels without nearby samples. To prevent aliasing, a bilinear interpolation was applied for real-time DEM construction by Johnson et al.~\cite{johnson2002LidarBasedHazardAvoidance}, which serves as a baseline DEM construction algorithm in this work. 

With the binary interpolation, each sample affects the four closest cells of the DEM. Let $\tilde{z}$ be the measured elevation of a sample point at 
the image coordinate $(u, v)$. Then, the four closest cells are $(u_i, v_i), (u_i, v_{i+1}), (u_{i+1}, v_{i})$, and $(u_{i+1}, v_{i+1})$ where $u_i \leq u < u_{i+1}$, $v_i \leq v < v_{i+1}$, and $u_i, u_{i+1}, v_{i+1}, v_i\in\mathbbm{Z}_{\geq0}$. For each sample, we update the two matrices $E\in\R^{n_D \times n_D}$ and $W\in\R^{n_D \times n_D}$ as follows.
\begin{equation}\label{eq:pcd2dem-johnson2002}
    \begin{aligned}
        & p = u - u_i, \quad & q = v - v_i \\
        & E(u_{i}, v_{i}) = (1-p)(1-q) z, \quad & W(u_{i}, v_{i}) = (1-p)(1-q) \\
        & E(u_{i+1}, v_{i}) = p(1-q) z, \quad & W(u_{i+1}, v_{i}) = p(1-q) \\
        & E(u_{i}, v_{i+1}) = (1-p)q z, \quad & W(u_{i}, v_{i+1}) = (1-p)q \\
        & E(u_{i+1}, v_{i+1}) = pq z, \quad & W(u_{i+1}, v_{i+1}) = pq
    \end{aligned}
\end{equation}
Figure \ref{fig:pcd2dem-johnson2002} shows that the weights in Eq. \ref{eq:pcd2dem-johnson2002} correspond to the area of each subsection located diagonally. For example, the weight for the pixel $(u_i, v_i)$ corresponds to the area of the subsection C.

After all samples have been accumulated, the DEM $D$ is obtained by the following element-wise computation.
\begin{equation}
    D(u, v) = E(u, v) / W(u, v)
\end{equation}

\begin{figure}[th]
    \centering
    \includegraphics[width=0.6\linewidth]{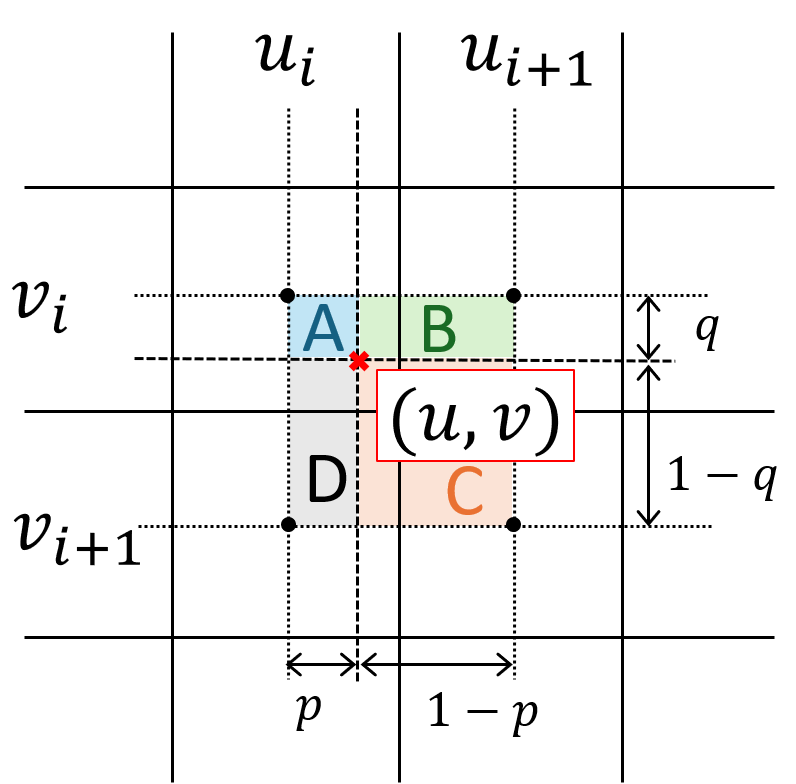}
    \caption{The diagram of the bilinear interpolation for DEM construction. The sampled data at $(u, v)$ is referenced for elevations at the closest four cells.}
    \label{fig:pcd2dem-johnson2002}
\end{figure}

Because of the occlusion with rocks or the scanner's irregular sampling, it is possible that some cells do not receive any PCD sample, resulting in holes in the DEM. For efficiency, the DEMs should be free of holes, and they are interpolated by assigning the average elevation of all neighboring cells that have elevation values. This process is repeated until all holes are filled~\cite{johnson2002LidarBasedHazardAvoidance}.

\subsection{Gaussian Random Field Regression}

The conventional method for DEM construction faces a significant limitation: the DEM resolution is determined based on the GSD of the PCD. This poses a challenge for autonomous planetary landing, where avoiding surface obstacles as small as 10 cm is necessary using DEMs obtained at a distance as far as 1000 meters.
Developing and operating LiDAR sensors with sufficient performance requires a sufficient budget and increases the lander's SWaP requirements.

To overcome this limitation, we propose constructing a \textit{Gaussian DEM} with the GRF model. A Gaussian DEM is a stochastic DEM where each cell stores the mean and variance of the elevation. With the GRF regression, the elevations at empty cells without any sample are estimated with associated uncertainties, which enables the construction of a dense but stochastic DEM. The DEM resolution can be set arbitrarily irrespective of the GSD.

\subsubsection{Measurement Model Assumed by GRF}
Let $f(\cdot):\R^2\rightarrow\R$ represent the true terrain topography, which is a bijective function that maps the horizontal coordinate $\bm{\gamma}=[x, y]^T\in\R^2$ to the local elevation. Suppose the terrain sensing provides point cloud data (PCD), denoted by $\mathcal{D} = \{(x_k, y_k, \tilde{z}_k) \mid k = 1, \ldots, n\}$, where $x_k, y_k, \tilde{z}_k$ are the 3D coordinates of the $k$-th sample, and $n$ is the total number of samples. Let $\mathcal{K}$ be the set of indices for the PCD samples. Then, the PCD can be represented as $\mathcal{D}=\{\bm{\Gamma}, \mathbf{\tilde{z}}\}$ where $\bm{\Gamma}=\{\bm{\gamma}_k\}_{k\in\mathcal{K}}$ is the set of horizontal coordinates and $\mathbf{\tilde{z}}=\{\tilde{z}_k\}_{k\in\mathcal{K}}$ is the set of the elevations. 

Application of the GRF means that the terrain sensing measurement noise is modeled as additive Gaussian noise to the elevation with a standard deviation of $\sigma_\epsilon$. 
\begin{equation}\label{eq: obs}
    \tilde{z}_{k} = f(\bm{\gamma}_{k}) + \epsilon, \quad \epsilon\sim\sN(0, \sigma_\epsilon^2)
\end{equation}
Note that the horizontal coordinate $\bm{\gamma}$ is an implicit function of the local topography and the terrain sensing system. 

\begin{remark}
The measurement error of real terrain sensing systems may not be as simple as Eq. \eqref{eq: obs}. For example, lidar measurement would be affected by the range measurement error, the laser pointing error, the vehicle navigation error, surface incident angle, or surface reflectivity, to name a few. However, given a continuous terrain surface assumption, the majority of the integrated measurement error could be represented as the elevation error. For specific hardware for particular flight missions, dedicated hardware experiments are required.
\end{remark}

\begin{remark}
The measurement error model introduced here is what the GRF model is based on, but it does not necessarily represent the noise model employed in the numerical demonstration. For instance, the LiDAR measurements in the numerical demonstration use additive Gaussian noise on each ray direction, instead of the elevation direction (z-axis). 
\end{remark}

\subsubsection{GRF-based Topographic Uncertainty}
The GRF regression, also called two-dimensional Gaussian process regression, fundamentally operates on the principle of assigning greater correlation to proximal points and lesser correlation to distant ones. This correlation is governed by a parameterized covariance function, also known as the kernel function, which dictates our assumptions about the terrain's characteristics, such as its differentiability and stationarity. For the purpose of modeling natural terrains, we employ the Absolute Exponential (AE) kernel function, defined as:
\begin{equation}\label{eq:ae-kernel-realtimeshd}
    k(\bm{\gamma}_i, \bm{\gamma}_j) = \sigma_f^2 \exp\left(-\frac{\|\bm{\gamma}_i - \bm{\gamma}_j\|}{\ell}\right),
\end{equation}
where $\|\cdot\|$ denotes the Euclidean norm. Here, $\ell$ is the length scale, influencing the radius within which elevation information is significantly correlated, and $\sigma_f$ represents the prior standard deviation, indicating the terrain's overall variability. Natural terrains are known to be well approximated by fractals, and the absolute exponential kernel function corresponds to a subset of the fractal Brownian motion~\cite{malamud2001WaveletAnalysesMars, turcotte1987fractal}

With the GRF model, we can compute the mean and variance of elevation at any arbitrary horizontal coordinate $\bm{\gamma}_*$, conditioned on observations $\mathcal{D} = \{\bm{\Gamma}, \bm{z}\}$, for a given set of hyperparameters, $\sigma_f$ and $\ell$~\cite{williams2006gaussian}:
\begin{align}
    \mu(\bm{\gamma}_*) &= \bm{k}_{*}^T\left(\bm{K}^{\epsilon}\right)^{-1}\mathbf{\tilde{z}} \label{eq:gpr-regular-mean}\\
    \sigma^2(\bm{\gamma}_*) &= k_{**} - \bm{k}_{*}^T\left(\bm{K}^{\epsilon}\right)^{-1}\bm{k}_*  \label{eq:gpr-regular-var}
\end{align}
where $\bm{K}^{\epsilon} = \bm{K} + \sigma_\epsilon^2\bm{I}_n$ and $\bm{K}\in\R^{n \times n}$ is the covariance matrix between the observed data points; $(\bm{K})_{ij} = k(\bm{\gamma}_i, \bm{\gamma}_j)$, and $i, j \in \mathcal{K}$. 
Similarly, $\bm{k}_{*}\in\R^n$ denotes the covariance between the prediction point $\bm{\gamma}_*$ and the observed data points $\bm{\Gamma}$, and $k_{**}=k(\bm{\gamma_*}, \bm{\gamma_*})=\sigma_f^2\in\R$ is the prior covariance of the prediction point. 

The major computational challenge in applying GRF for Gaussian DEM construction lies in the matrix inversion required by these equations, which typically scales as $O(n^3)$ where $n$ is the number of sampled points. Our proposed methodology, as detailed in the following subsection, addresses this computational bottleneck, making GRF a viable option for real-time Gaussian DEM construction.

\subsection{Local GRF Regression via Delaunay Triangulation}
The fundamental idea for real-time Gaussian DEM construction involves local approximation, particularly through Delaunay triangulation. Delaunay triangulation enables the rapid division of terrain into triangular segments, utilizing vertices from the PCD obtained through terrain sensing. Applying GRF to each triangle independently allows for an efficient estimation of terrain topography. Note that local approximation is an effective strategy for assessing landing safety, which is inherently dependent on the local terrain characteristics within the diameter of the lander at each potential landing site.

Triangulation is a classical geometric structure in computational geometry that partitions a region spanned by points into triangular segments. In particular, Delaunay triangulation avoids skinny triangles by maximizing the minimum angle over all triangulations, and is commonly used to approximate natural terrain from sparse datasets. For example, a Triangulated Irregular Network (TIN), a form of digital terrain model is often based on Delaunay triangulation.

Delaunay triangulation is known for its efficiency, typically running in $O(n\log n)$ expected time using the incremental insertion algorithm. There are more efficient algorithms, such as the divide-and-conquer algorithm of Lee and Schachter \cite{lee1980two} and the plane-sweep algorithm of Fortune. In the numerical experiments of this study, we used a C program library called Triangle \cite{shewchuk1996triangle}, which implements an efficient Delaunay triangulation algorithm utilizing the divide-and-conquer with alternating cuts. 

Given the triangulation, we infer the terrain topography by the GRF regression to each triangular segment independently. Specifically, we apply Eqs. \eqref{eq:gpr-regular-mean} and \eqref{eq:gpr-regular-var} to estimate the elevation inside the triangle based on the three vertices. Now, the matrix $\bm{K}^{\epsilon}$ is reduced to a 3x3 size, allowing for the analytical computation of its inverse and accelerating the process. Figure \ref{fig:pcd2dem-proposed} shows the illustration of the local GRF regression via Delaunay triangulation. For example, elevations within the orange triangle in Fig.~\ref{fig:pcd2dem-proposed} are estimated using GRF regressions with the elevations of the triangle's vertices, shown by the red dots.

\begin{figure}[htb]
    \centering
    \includegraphics[width=0.75\linewidth]{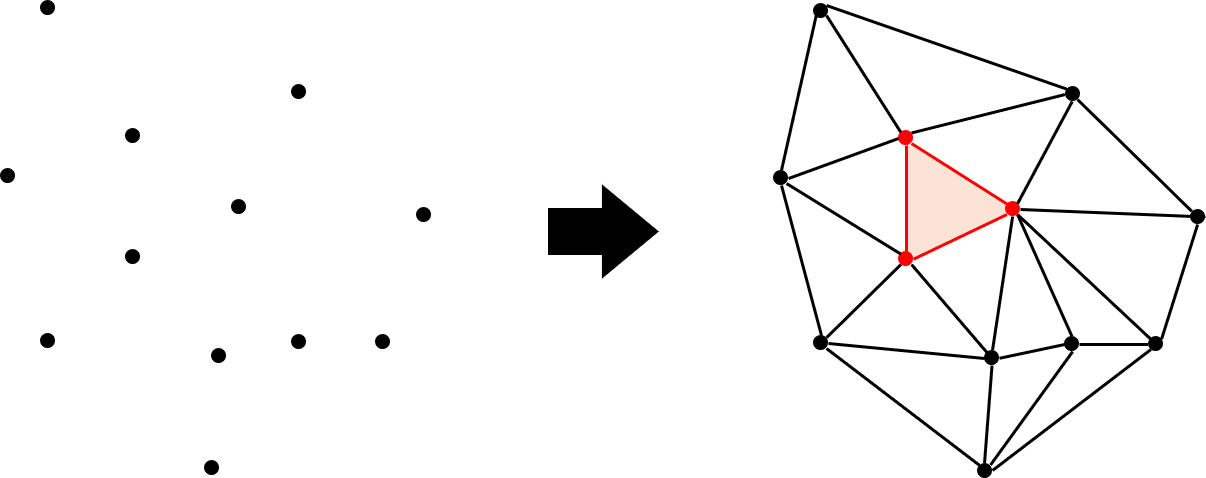}
    \caption{The mean and variance of elevations are locally interpolated using GRF regressions with the three nearby samples identified by Delaunay triangulation.}
    \label{fig:pcd2dem-proposed}
\end{figure}

\subsection{Real-Time Gaussian DEM Construction Algorithm}
For efficiency, the process begins by running the conventional DEM construction algorithm reviewed in Section \ref{sec:conventional-pcd2dem} but with a sufficiently high resolution defined by the user. This way, we can avoid unnecessary triangulation and GRF regression over sufficiently sampled regions. Then, Delaunay triangulation is performed, followed by the local GRF regression. The pixels within each triangle are identified using Bresenham's line algorithm, an efficient method for drawing lines in raster data.

\section{Real-Time Safety Safety Evaluation}\label{landing-safety}
This section covers the landing safety evaluation based on DEMs. First, we review the definition of landing safety, followed by the introduction of the baseline algorithm, which is based on the SOTA HD algorithm from NASA's ALHAT program. Then, we present a novel deterministic real-time HD algorithm, and its guaranteed conservativeness is shown. Finally, the proposed deterministic HD algorithm is extended to a stochastic setting so that it can accommodate Gaussian DEM input.

\subsection{Definition of Landing Safety}
Ivanov et al.~\cite{ivanov2013ProbabilisticHazardDetection} introduced the formal definition of landing safety based on two fundamental factors: the slope and roughness, which is the standard way of evaluating landing safety for planetary missions. 
Both slope and roughness are defined by the geometric relationship between the local terrain and the \textit{landing plane}, which is the plane formed by the landing legs in contact with the terrain (Fig. \ref{fig:slope_rghns_def}). The slope is the tilting angle of the surface normal vector of the landing plane with respect to the local gravitational vector, while the roughness is the distance between the landing surface and surface obstacles. A target site is deemed safe if both the slope and roughness fall below predetermined thresholds across all lander orientation angles, $\theta$.

Formally, let $\bm{p}_1, \bm{p}_2, \bm{p}_3 \in \mathbb{R}^3$ be the positions of the three landing legs in contact with the terrain at orientation angle $\theta$, which are ordered in counterclockwise looked from above. Then, the surface normal vector of the landing plane, i.e., the plane spanned by $\bm{p}_1, \bm{p}_2, \bm{p}_3$, is derived from the cross product:
\begin{equation}\label{eq:normal}
    \bm{n}(\theta) = (\bm{p}_2 - \bm{p}_1) \times (\bm{p}_3 - \bm{p}_1) 
             \triangleq [a, b, c]^T.
\end{equation}
The slope, $s(\theta)$, is the angle between the surface's normal vector $n$ and the vertical axis:
\begin{equation}\label{eq: slope}
    s(\theta) = \arccos{\left(\frac{c}{\sqrt{a^2 + b^2 + c^2}}\right)}.
\end{equation}
The roughness $r(\theta, \bm{\gamma})$ represents the signed distance between the terrain and the landing surface at any horizontal location $\bm{\gamma} \in U(\theta)$, where $U(\theta)$ is the set of horizontal coordinates within the lander footprint:
\begin{equation}\label{eq:rghns}
    r(\theta, \bm{\gamma}) = \frac{ax + by + cz + d}{\sqrt{a^2 + b^2 + c^2}}, \quad \bm{\gamma} = (x, y) \in U(\theta), \quad z = f(\bm{\gamma}),
\end{equation}
where $z = f(\bm{\gamma})$ is the terrain elevation at $\bm{\gamma}$, and $d$ is determined such that $ax_i + by_i + cz_i + d = 0$ for any $\bm{p}_i = [x_i, y_i, z_i]^T$, $i = 1, 2,$ and $3$. We assume $f(\cdot)$ is bijective; thus, $c$ is non-zero and positive, i.e., $c> 0$. 

\begin{definition}[Landing safety]
Let $\bar{s}>0$ and $\bar{r}>0$ be the slope and roughness thresholds, respectively. Then, a candidate landing site is a safe landing site if and only if the slope and roughness are less than the thresholds for all possible orientations, i.e., $s(\theta) < \bar{s}$ and $r(\theta, \bm{\gamma}) < \bar{r}, \forall\bm{\gamma}\in U(\theta), \forall \theta\in[-\pi, \pi]$. 
\end{definition}

\subsection{Baseline HD Algorithm}
The baseline HD algorithm in this study is based on the ALHAT HD algorithm~\cite{ivanov2013ProbabilisticHazardDetection, johnson2022OPTIMIZATIONLIDARBASED}.
The ALHAT HD algorithm takes a DEM and computes the deterministic slope safety and the stochastic roughness safety. First, it evaluates the landing leg's elevation across the DEM by taking the maximum elevation of the landing leg footprint and stores the result into a \textit{footpad map}. Given the footpad map, the slope is evaluated deterministically for all possible lander orientations across the DEM. For each orientation at each DEM pixel, the landing plane is defined and used to compute the roughness safety.

The ALHAT HD algorithm uses the maximum elevation within the lander footprint for the roughness safety evaluation for efficiency and claims this approach guarantees the conservative roughness evaluation~\cite{johnson2022OPTIMIZATIONLIDARBASED}. Without this approximation, the probability of roughness safety must be computed for all pixels within the lander footprint for all orientations across the DEM, which is computationally expensive. Furthermore, the ALHAT HD algorithm assumes the independent and identically distributed (i.i.d.) Gaussian elevation noise at each DEM pixel, which makes the joint probability of roughness safety almost always near zero, especially when there are many pixels within the lander footprint~\cite{ivanov2013ProbabilisticHazardDetection}. Therefore, it is practical to use a single representative value for roughness evaluation for each orientation across the DEM.

However, using the maximum elevation within the lander footprint does not guarantee the conservative roughness evaluation. The roughness is defined by the distance between the local elevation and the landing plane. Unless the landing plane is completely flat, i.e., perpendicular to the vertical axis, the separation distance from the terrain depends on the horizontal coordinate. This is also immediate from the roughness equation of Eq. \eqref{eq:rghns}. Figure \ref{fig:counter-example} shows a counter-example; rock A has a higher elevation but results in smaller roughness than rock B, i.e., $z_A > z_B$ but $r(\theta, \gamma_A)<r(\theta, \gamma_B)$. 
\begin{figure}
    \centering
    \includegraphics[width=0.7\linewidth]{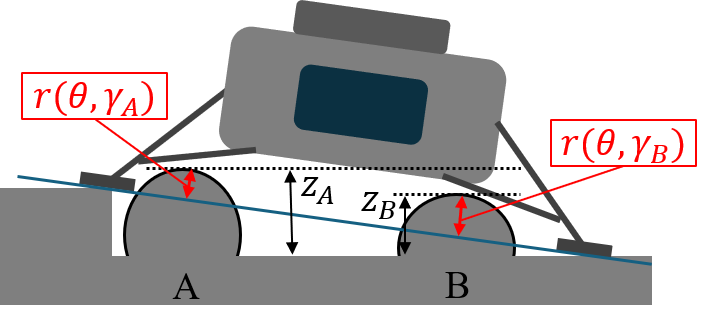}
    \caption{The highest local elevation does not necessarily correspond to the worst roughness; rock A has a higher elevation but results in smaller roughness than rock B. }
    \label{fig:counter-example}
\end{figure}

Therefore, in our baseline HD algorithm, the probability of roughness safety is evaluated for all pixels within the lander footprint for each orientation across the DEM. Given the number of pixels for a unit length, denoted by $p$, the algorithm's complexity is $O(p^5)$, which is still the same as the ALHAT HD algorithm~\cite{johnson2022OPTIMIZATIONLIDARBASED}. For the scalability of HD algorithms, see Appendix \ref{appdx: scalability} as well.

\subsection{Provably Conservative Real-Time Deterministic HD Algorithm}
Here, we present a novel real-time deterministic HD algorithm with guaranteed conservativeness. For efficiency, the proposed method leverages the height difference-based evaluation of slope and roughness. It also eliminates the need for repetitive evaluations across all possible orientations through simultaneous consideration of potential landing leg locations and lander footprint.  

Given a candidate landing site, let $L$ be the set of all the horizontal coordinates where all the landing legs' footprints would overlap; $L=\cup_{i}\cup_{\theta}L_i(\theta)$ where $L_i(\theta)$ is a set of horizontal coordinates the landing leg $i$ overlaps for a given orientation $\theta$.
Similarly, let $U=\cup_{\theta}U(\theta)$ denote the union of the possible lander footprint for all orientation angles. Figure \ref{fig:lander-geom-realtime} shows the footprints of landing legs and the lander for an orientation $\theta$, and their union for all orientations.
Then, we define the maximum and minimum of the elevation over these sets as follows:
\begin{equation}
\begin{split}
    \overline{z}_L := \max_{\bm{\gamma}\in L}f(\gamma), \quad
    \underline{z}_L := \min_{\bm{\gamma}\in L}f(\gamma)\\
    \overline{z}_U := \max_{\bm{\gamma}\in U}f(\gamma), \quad
    \underline{z}_U := \min_{\bm{\gamma}\in U}f(\gamma).
\end{split}
\end{equation}
\begin{figure}
    \centering
    \includegraphics[width=1.0\linewidth]{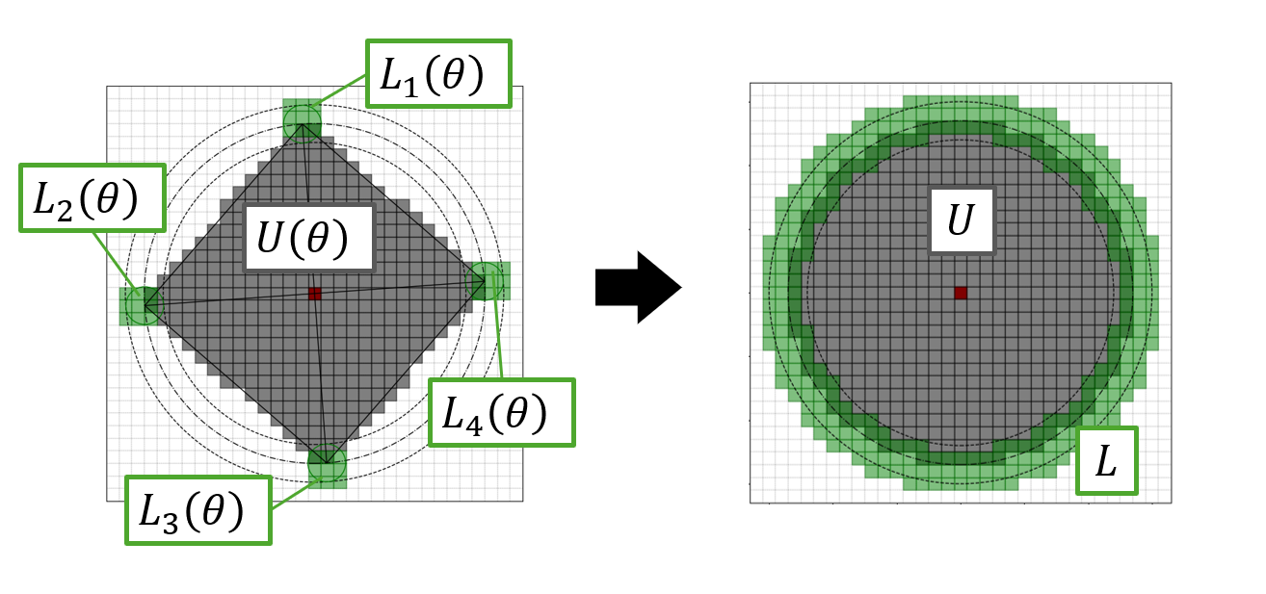}
    \caption{Footprints of landing legs and the lander for an orientation $\theta$, and their union for all orientations.}
    \label{fig:lander-geom-realtime}
\end{figure}

For all the possible triangles formed by any three landing legs, let $h_0$ be the minimum height of such triangles; the distance from a vertex to the opposite edge. Now, we define conservative slope $\Tilde{s}$ and conservative roughness $\Tilde{r}$ as follows. 
\begin{equation}\label{eq: conservative-slope-rghns}
    \Tilde{s}:= \sin^{-1}\left({\frac{\overline{z}_L - \underline{z}_L}{h_0}}\right), \quad \Tilde{r} := \overline{z}_U - \underline{z}_L
\end{equation}

The real-time deterministic HD algorithm returns safe if a candidate landing site exhibits $\overline{z}_L - \underline{z}_L < h_0 \sin(\bar{s})$ and $\overline{z}_U - \underline{z}_L < \bar{r}$. Note that the value $h_0 \sin(\bar{s})$ should be precomputed to avoid costly floating computations in realtime. 

The proposed deterministic HD algorithm guarantees conservative evaluations of slope and roughness if the lander footprint is inscribed within the polygon formed by landing legs. The assumption and the conservative evaluation are formally stated as follows.

\begin{assumption}\label{assumption:inscribed-lander-footprint}
The lander's horizontal footprint, across which the terrain's roughness metric is evaluated, is inscribed within the horizontal projection of the two-dimensional polygon formed by connecting all of a lander's landing legs.
\end{assumption}

\begin{proposition}\label{proposition:dhd-conservative}
Under Assumption \ref{assumption:inscribed-lander-footprint}, a candidate landing site is safe if conservative slope and roughness are smaller than the threshold, i.e., $\Tilde{s} < \bar{s}$ and $\Tilde{r} < \bar{r}$. 
\end{proposition}

To prove Proposition \ref{proposition:dhd-conservative}, we first introduce the following theorem about the slope of a triangle confined between two horizontal planes as shown in Fig. \ref{fig:tri-btw-planes}.
\begin{figure}
    \centering
    \includegraphics[width=0.75\linewidth]{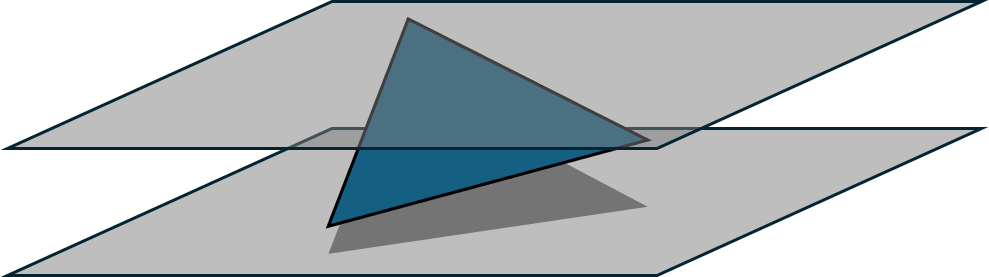}
    \caption{A triangle confined between two horizontal planes.}
    \label{fig:tri-btw-planes}
\end{figure}
\begin{theorem}\label{theorem:max-slope-tri}
Consider a triangle with a free orientation. Let $z_i$ be the elevation of a vertex $i$ and $h_i$ be the distance from a vertex $i$ to the opposite edge, and $i=1, 2, 3$. Suppose the elevations of vertices are uniformly bounded by a constant that is smaller than any $h_i$, i.e., $\exists \Delta z$ s.t. $0\leq z_i \leq \Delta z < h_i$, $\forall i$. If the triangle's slope is maximized under this constraint, then $z_i=0$ or $z_i=\Delta z$ for all $i$, but not all the elevations are the same. Furthermore, the maximum slope is $\sin^{-1}(\Delta z/h_0)$ where $h_0\triangleq \min_i h_i$. 
\end{theorem}
\begin{proof}
See Appendix.
\end{proof}

With Theorem \ref{theorem:max-slope-tri}, we proceed to prove Proposition \ref{proposition:dhd-conservative}. 

\begin{proof}[Proof of Proposition \ref{proposition:dhd-conservative}]
We show that, under Assumption \ref{assumption:inscribed-lander-footprint}, the conservative slope $\Tilde{s}$ and roughness $\Tilde{r}$ is larger than or equal to the exact slope $s$ and roughness $r$ of Eqs. \eqref{eq: slope} and \eqref{eq:rghns}, respectively. The two-dimensional schematic for the proof is given by Fig. \ref{fig:proof-concept}.

For an arbitrary lander orientation $\theta$, let $r_z(\theta, \bm{\gamma})$ be the signed distance from a point on the terrain, $(\bm{\gamma}=(x, y), z))$, to its projection in z-direction onto the landing plane:
\begin{equation}\label{eq: vertical-dist}
    r_z(\theta, \bm{\gamma}) = z - \left(-\frac{ax + by + d}{c}\right) = \frac{ax + by +cz + d}{c}.
\end{equation}

For an arbitrary lander orientation $\theta$, the minimum elevation of all the landing legs, $\underline{z}_L(\theta):= \min_{\bm{\gamma}\in \cup_i L_i(\theta)}f(\bm{\gamma})$, is the minimum elevation on the two-dimensional polygon $P_L(\theta)$ formed by the landing legs because the vertices of this polygon are the landing legs. The polygon $P_L(\theta)$ exists on the landing plane and, with Assumption \ref{assumption:inscribed-lander-footprint}, the maximum vertical signed distance between the terrain and the landing plane across lander footprint is thus bounded as follows.
\begin{equation}
    \max_{\bm{\gamma}\in U(\theta)} r_z(\theta, \bm{\gamma}) \leq \overline{z}_U(\theta) - \underline{z}_L(\theta)
\end{equation}
where $\overline{z}_U(\theta):=\max_{\bm{\gamma}\in U(\theta)}f(\bm{\gamma})$.

Since $c>0$, we have $r_z(\theta, \bm{\gamma}) \geq r(\theta, \bm{\gamma})$ for any $\theta$. Thus, we obtain $r(\theta, \bm{\gamma}) \leq \overline{z}_U(\theta) - \underline{z}_L(\theta)$ for any $\theta$, which follows that $r(\theta, \bm{\gamma}) \leq \overline{z}_U - \underline{z}_L=\tilde{r}, \forall\bm{\gamma}\in U(\theta), \forall\theta$ because $\overline{z}_U \geq \overline{z}_U(\theta)$ and $\underline{z}_L \leq \underline{z}_L(\theta)$.

For slope, let $\theta^*$ be the orientation resulting in the maximum slope $s^*$ and $\mathcal{I}$ be the three landing leg indices that are in contact with the terrain at $\theta^*$. We denote the maximum and minimum elevations of the landing legs $\mathcal{I}$ as follows.
\begin{equation}
    \overline{z}_L(\theta^*) := \max_{\bm{\gamma}\in \cup_{i\in\mathcal{I}}L_i(\theta^*)} f(\bm{\gamma}), \quad
    \underline{z}_L(\theta^*) := \min_{\bm{\gamma}\in \cup_{i\in\mathcal{I}}L_i(\theta^*)} f(\bm{\gamma})
\end{equation}
Then, by Theorem \ref{theorem:max-slope-tri}, the maximum slope is expressed as $s^*=\sin^{-1}(\Delta z(\theta^*)/h)$ where $\Delta z(\theta^*)=\overline{z}_L(\theta^*) - \underline{z}_L(\theta^*)$ and $h\in\{h_i\}_{i\in\mathcal{I}}$. Therefore, using $\overline{z}_L\geq \overline{z}_L(\theta^*)$, $\underline{z}_L\leq \underline{z}_L(\theta^*)$, and $h_0\leq h$, we obtain $\tilde{s} > s^*$. 
\end{proof}

\begin{figure}
    \centering
    \includegraphics[width=0.9\linewidth]{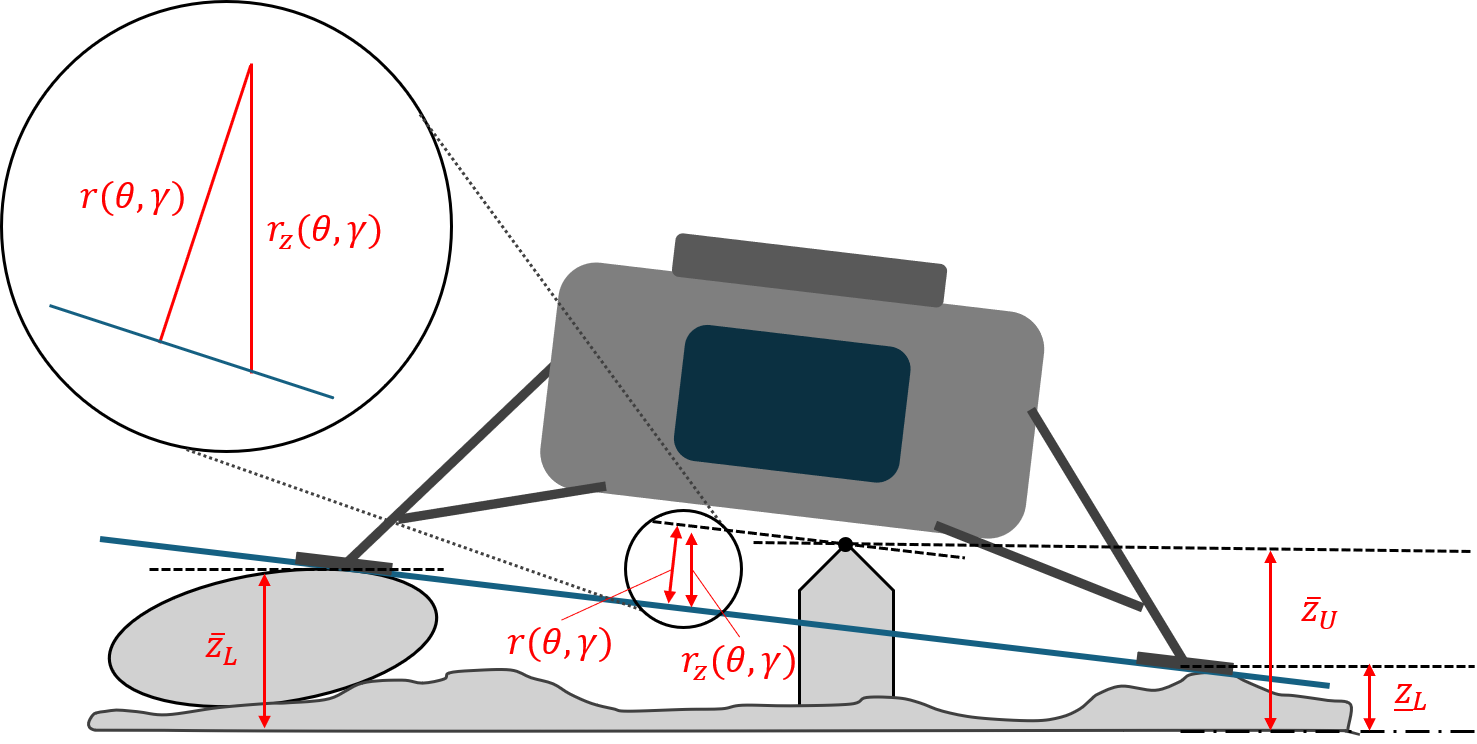}
    \caption{A two-dimensional schematic for the proof of Proposition \ref{proposition:dhd-conservative}.}
    \label{fig:proof-concept}
\end{figure}

\subsection{Extension to Gaussian DEM Input}
In a stochastic setting with the Gaussian DEM, computing the distribution of the maximum or the minimum value of Gaussian random variables is not straightforward. In this study, we use a heuristic approach based on the $3-\sigma$ bounds. Suppose we have independent Gaussian random variables $X_1, ..., X_m$, where $X_i \sim\sN(\mu_i, \sigma_i^2)$. Then, we compute the following bounds.
\begin{equation}
\begin{split}
    &\overline{c}_{\max} = \max_i \mu_i + 3 \sigma_i, \quad \underline{c}_{\max} = \max_i \mu_i - 3 \sigma_i \\
    &\overline{c}_{\min} = \min_i \mu_i + 3 \sigma_i, \quad \underline{c}_{\min} = \min_i \mu_i - 3 \sigma_i
\end{split}
\end{equation}
Then, we approximate the maximum and minimum of $\{X_i\}_i$ as follows (also see Fig. \ref{fig:3sigma-approx}). 
\begin{equation}
\begin{split}
    &\max_i X_i \approx \overline{X}\sim\sN(\mu_{\max}, \sigma_{\max}^2), \quad
    \mu_{\max} = \frac{\overline{c}_{\max} + \underline{c}_{\max}}{2}, \quad
    \sigma_{\max} = \frac{\overline{c}_{\max} - \underline{c}_{\max}}{6}\\
    &\min_i X_i \approx \underline{X}\sim\sN(\mu_{\min}, \sigma_{\min}^2), \quad
    \mu_{\min} = \frac{\overline{c}_{\min} + \underline{c}_{\min}}{2}, \quad
     \sigma_{\min} = \frac{\overline{c}_{\min} - \underline{c}_{\min}}{6}
\end{split}
\end{equation}
With this approach, we can compute the Gaussian approximations of $\overline{z}_L, \underline{z}_L$, and $\overline{z}_U$. Since the distribution of a difference of two normally distributed variables $X\sim\sN(\mu_x, \sigma_x^2)$ and $Y\sim\sN(\mu_y, \sigma_y^2)$ is another normal distribution as $X - Y \sim(\mu_x - \mu_y, \sigma_x^2 + \sigma_y^2)$, computation of the approximated Gaussian distributions of $\overline{z}_L - \underline{z}_L$ and $\overline{z}_U - \underline{z}_L$ is straightforward. Let $\sN(\mu_{LL}, \sigma_{LL}^2)$ and $\sN(\mu_{UL}, \sigma_{UL}^2)$ represent the distributions of $z_{LL}=\overline{z}_L - \underline{z}_L$ and $z_{UL}=\overline{z}_U - \underline{z}_L$, respectively. 
Then, the probability of slope safety and roughness safety are computed as follows:
\begin{equation}
    \mathbbm{P}(\tilde{s} < \bar{s}) \approx \Phi\left(\frac{h_0 \tan(\bar{s}) - \mu_{LL}}{\sigma_{LL}}\right), \quad \mathbbm{P}(\tilde{r} < \bar{r}) \approx \Phi\left(\frac{\bar{r} - \mu_{UL}}{\sigma_{UL}}\right)
\end{equation}
where $\Phi$ is a cumulative distribution function for the standard normal distribution. 
\begin{figure}
    \centering
    \includegraphics[width=0.5\linewidth]{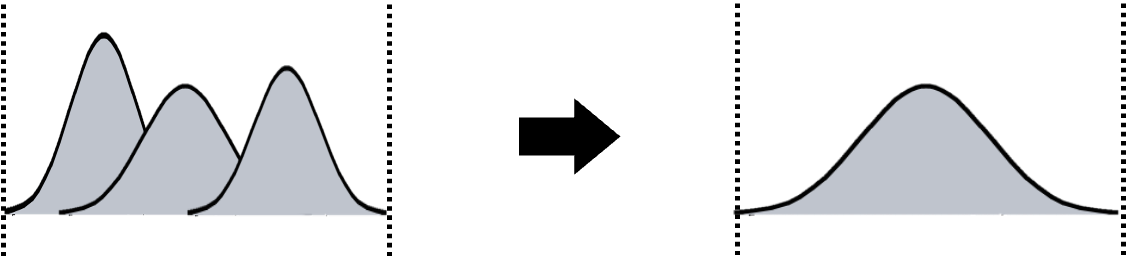}
    \caption{Illustration of the Gaussian approximation of the max/min elevation over a given footprint.}
    \label{fig:3sigma-approx}
\end{figure}

\subsection{Implementation}
Computing the exact locations of landing legs for a given target terrain and orientation requires simulating the terrain-lander physical interactions, which is computationally too expensive to apply across all locations. Since the slope thresholds for planetary landers are relatively small, e.g., $5\sim15$ degrees, the landing legs' horizontal coordinates for a given orientation are often approximated to be constant, irrespective of the terrain topography. 

For instance, let $\bm{l}_{i\theta}=(x_{i\theta}, y_{i\theta}, 0)$ be the position of a landing leg $i$ for an arbitrary lander orientation $\theta$ on flat terrain. Here, we consider a local-vertical local-horizon (LVLH) frame with the origin at the nadir of the lander center. If the terrain has a constant slope $\alpha$ along x-axis such that 
\begin{equation}
    f(x, y) = x \tan(\alpha), 
\end{equation}
then the exact location of the landing leg is $\bm{l}_{i\theta}=[x_{i\theta} \cos(\alpha), y_{i\theta}, x_{i\theta}\sin(\alpha)]^T$. However, we approximate it with the constant horizontal coordinates as follows:
\begin{equation}\label{eq: flat-approx}
    \bm{l}_{i\theta}\approx [x_{i\theta}, y_{i\theta}, f(x_{i\theta}, y_{i\theta})]^T.
\end{equation}
With this approximation, given the lander geometry, we can pre-compute all the possible landing leg horizontal coordinates for all orientation angles.

\section{Numerical Analysis}\label{sec: realtime-shd-demo}
\subsection{Conservativeness Demonstration}

First, we demonstrate the conservativeness of the proposed HD algorithm using perfect observations. While we have theoretically proven the conservativeness of the algorithm, we also present numerical demonstrations. Here, we utilize real terrain information, specifically the Mars digital terrain model (DTM) obtained by the High Resolution Imaging Science Experiment (HiRISE) camera~\cite{hirise}.

The level of conservativeness depends on the terrain complexity. Therefore, we prepared a list of terrains with varying complexities. We created these digital elevation models (DEMs) by superimposing a rocky plane with scaled terrains from the Mars DTM.

Initially, we generate a rocky terrain by placing rocks of random sizes at random locations on a flat surface. We then load specific local terrain data from the Mars DTM. Let us denote the DEM for the rocky terrain as $D_\text{rock}$ and the DEM for the Mars DTM as $D_\text{terrain}$. These DEMs are elevation maps represented by 2D arrays, with each element representing the local elevation. We then superimpose them with the terrain complexity factor, $c > 0$, as follows:
\begin{equation}
D_c = D_\text{rock} + c D_\text{terrain}
\end{equation}
By using different values of $c$, we can create a list of DEMs with various terrain complexities. Here, we set DEM resolution to be $10$ cm/pix. The HiRISE DEMs are in $1$ m/pix resolution, so they are interpolated beforehand.  

We applied our proposed HD algorithms to each DEM. As performance metrics, we used precision and recall.
\begin{align}\label{eq: precision}
\textrm{Precision} = \frac{\textrm{True Safe}}{\textrm{True Safe + False Safe}}
\end{align}
\begin{align}\label{eq: recall}
\textrm{Recall} = \frac{\textrm{True Safe}}{\textrm{True Safe + False Unsafe}}
\end{align}
In our task of safety mapping for planetary landing, achieving a precision as close as possible to 1 is critical. Precision, as defined in Eq. \eqref{eq: precision}, represents the rate of truly safe pixels out of those predicted to be safe. Pixels predicted to be safe can always be selected as candidate landing sites, so maintaining high reliability is crucial. On the other hand, recall, as defined in Eq. \eqref{eq: recall}, indicates the number of safe pixels detected out of all safe sites. High recall means that the algorithm provides many candidate landing sites. Typically, precision and recall have a tradeoff relationship; achieving high reliability in safety prediction (high precision) often results in labeling many safe pixels as unsafe (low recall). Note that in the context of binary classification, \textit{true positive} or \textit{false negative} are often used. In our analysis, "positive" corresponds to safe pixels.

\begin{table}[htbp!]
\caption{\label{tab:lander-geom-safety} Lander geometry and landing safety parameters.}
\begin{center}
\begin{tabular}{ c c }
\toprule
{Parameter} & {Value}  \\ \hline \hline
Lander Diameter (m) & 5.0 \\  
Landing Pad Diameter (m) & 0.3 \\ 
Critical Slope (deg) & 10.0 \\ 
Critical Roughness (m) & 0.25 \\ \bottomrule
\end{tabular}
\end{center}
\end{table}

\begin{table}[h!]
\caption{\label{table:precision_recall_conservativeness} Precision and Recall values for different terrain complexities, denoted by $c$.}
\centering
\begin{tabular}{ccc}
\toprule
$c$ & {Precision} $\uparrow$ & {Recall} $\uparrow$ \\ \hline\hline
0.0 & 1.0000 & 0.9563 \\
0.2 & 1.0000 & 0.7251 \\
0.5 & 1.0000 & 0.2848 \\
0.7 & 1.0000 & 0.1604 \\
1.0 & 1.0000 & 0.0810 \\
\bottomrule
\end{tabular}
\end{table}

We applied our proposed HD algorithm with the lander geometry and critical slope and roughness values presented in Table \ref{tab:lander-geom-safety}. Table \ref{table:precision_recall_conservativeness} shows the precision and recall for safety prediction under perfect observation conditions.
Although the proposed algorithm misses more safe sites in more complex terrains, as indicated by smaller recall values for larger $c$, it achieves 100\% precision for all cases; the pixels predicted to be safe are always safe regardless of terrain complexity. This result numerically demonstrates the conservatives of the proposed HD algorithm. 

\begin{figure}
    \centering
    \includegraphics[width=1\linewidth]{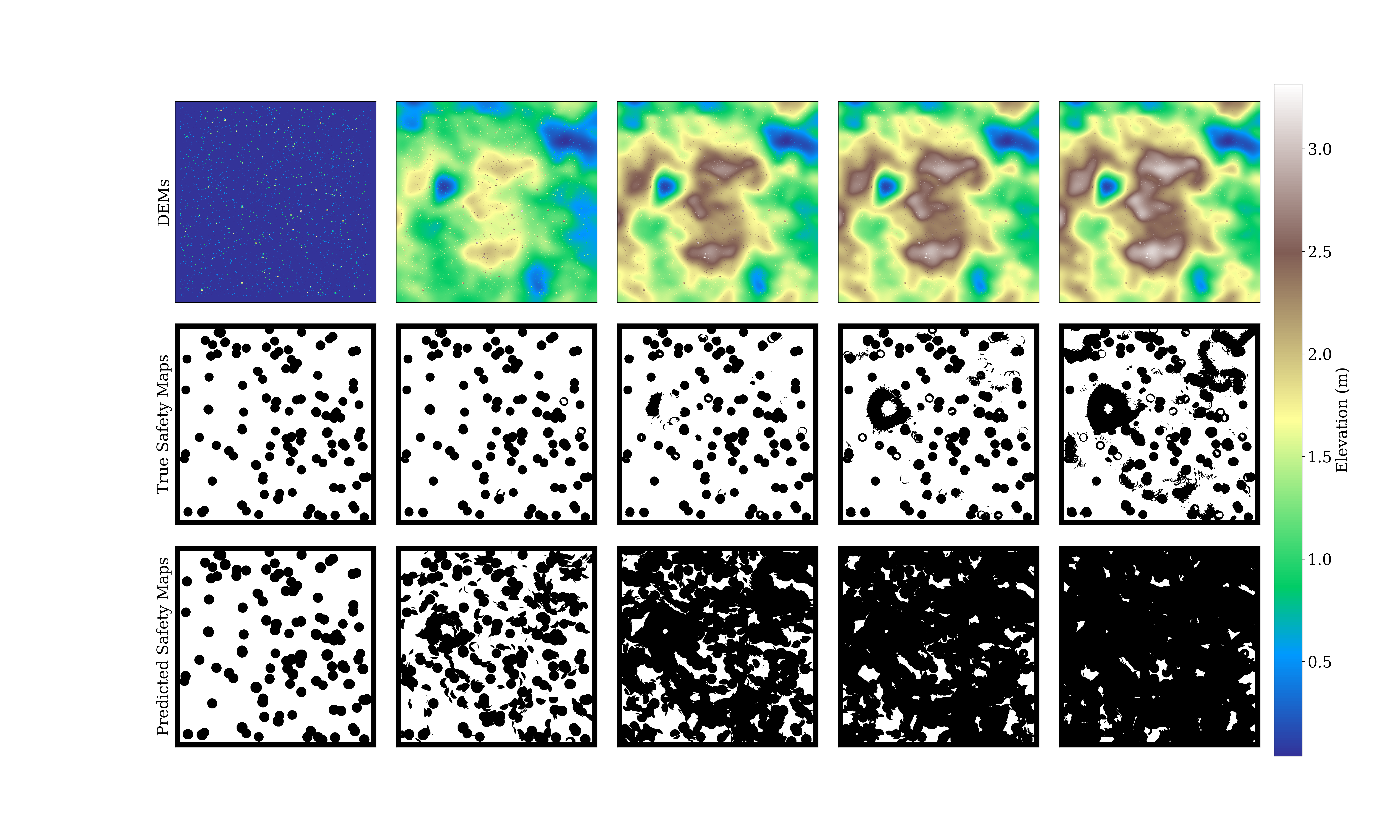}
    \caption{Visualization of safety prediction results across varying terrain complexities. White pixels are the safe pixels in the safety maps.}
    \label{fig:mylabel}
\end{figure}

Figure \ref{fig:mylabel} shows visualized results for varying terrain complexities; from left to right, terrain complexity increases and the map size is 100 meters by 100 meters in 10 cm/pixel resolution for all maps. The top row shows the DEM with perfect observation, and the second row shows the ground truth of landing safety. The third row shows the predicted safety maps by the proposed algorithm. We observe that the number of pixels predicted to be safe decreases as terrain complexity increases.

\subsection{Performance Analysis with Simulated Testbed}
In practice, obtaining perfect DEMs is challenging, if not impossible. When a planetary lander performs terrain sensing upon arrival at a planetary body, accurately measuring every detail of the surface is difficult. Specifically, we focus on topographic uncertainty due to sparse terrain measurements. In this subsection, we evaluate the performance of the proposed algorithm for such sparse terrain measurements using simulated testbed data. Here, as for the hyperparameters of Eq. \eqref{eq:ae-kernel-realtimeshd}, we set $\ell=1.0$ m and compute $\sigma_f$ as the standard deviation of the elevations in the point cloud data.

\subsubsection{Data Generation}

\begin{figure}
    \centering
    \includegraphics[width=1\linewidth]{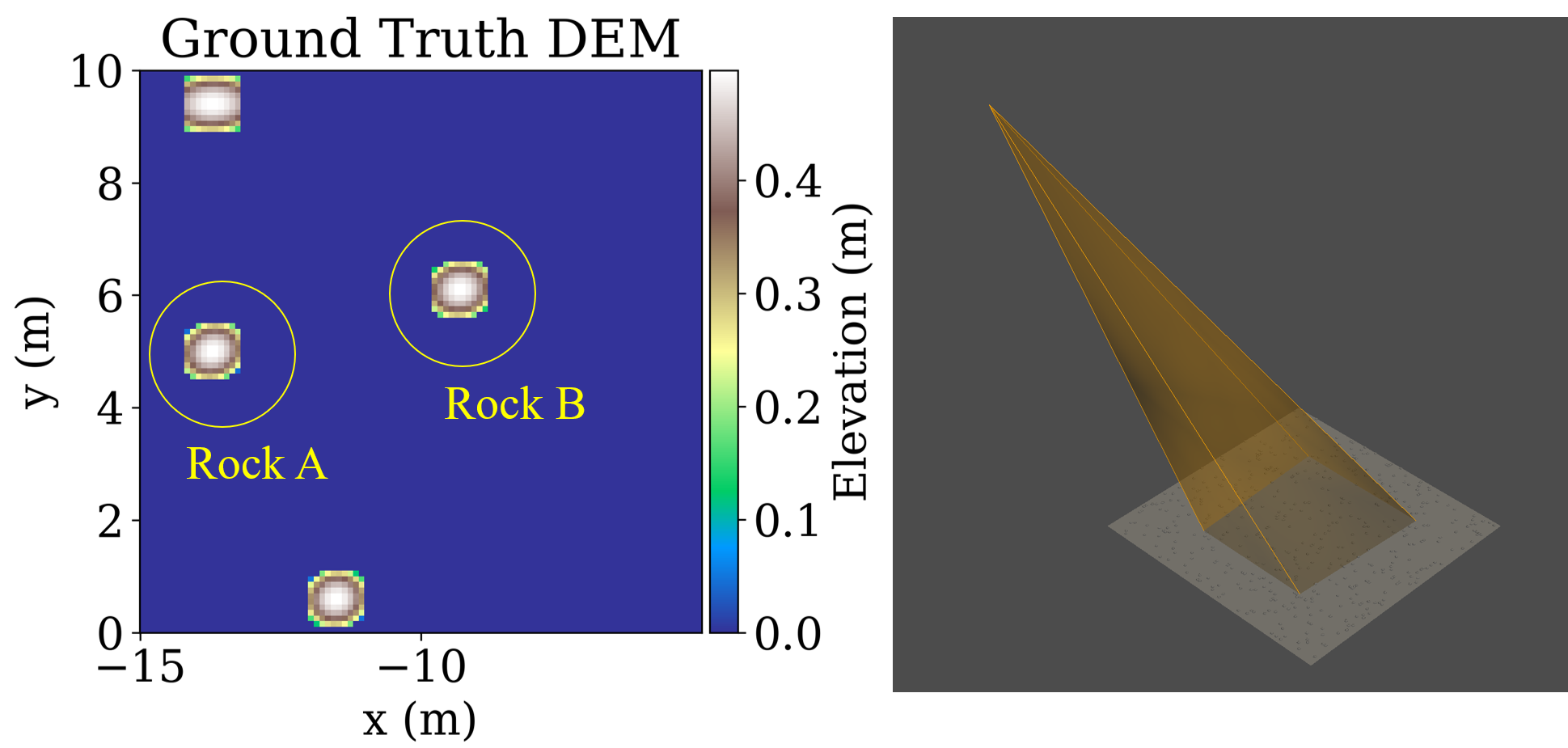}
    \caption{Left: Testbed DEM (zoomed in over a 10 x 10 meter segment). Right: Illustration of the simulated LiDAR scan.}
    \label{fig:rock-data-pipeline}
\end{figure}

For the simulated testbed, we placed 500 rocks, each with a diameter of 1 meter, over a flat 200x200 meter surface. Rocks were placed randomly, avoiding overlap. The rocks were modeled as hemi-ellipsoids with heights equal to half their semi-major axis. Figure \ref{fig:rock-data-pipeline} shows a sample z terrain and the LiDAR scan configuration.

\begin{table}[ht]
\caption{\label{tab:lidar-config} Nominal coverages, ground sample distances (GSDs), and Gaussian noise standard deviations for each LiDAR scan configuration$^*$.}
\centering
\begin{tabular}{ccccccc}
\toprule
\multicolumn{2}{c}{{Scan Configuration}} & \multicolumn{2}{c}{{Coverage, m}} & \multicolumn{2}{c}{{GSD, m}} & \multicolumn{1}{c}{{Noise}, cm} \\
Range, m & Angle, deg & x & y & x & y & $3\sigma_\text{LiDAR}$ \\
\hline \hline
200 & 0 & 40.00 & 40.00 & 0.156 & 0.156 & 2.0 \\
200 & 30 & 46.34 & 40.00 & 0.181 & 0.156 & 2.0 \\
200 & 60 & 82.47 & 40.00 & 0.322 & 0.156 & 2.0 \\
500 & 0 & 100.00 & 100.00 & 0.391 & 0.391 & 5.0 \\
500 & 30 & 115.86 & 100.00 & 0.453 & 0.391 & 5.0 \\
500 & 60 & 206.19 & 100.00 & 0.805 & 0.391 & 5.0 \\
1000 & 0 & 200.00 & 200.00 & 0.781 & 0.781 & 10.0 \\
1000 & 30 & 231.71 & 200.00 & 0.905 & 0.781 & 10.0 \\
1000 & 60 & 412.37 & 200.00 & 1.611 & 0.781 & 10.0 \\ \bottomrule
\end{tabular}
\begin{tablenotes}
\item {\footnotesize $^*$ The presented values are provided for reference purposes only; the simulated scans exhibit variations in coverage and GSDs due to terrain interactions and projection distortions.}
\end{tablenotes}
\end{table}
Simulated LiDAR scans were performed with various range and angle configurations using a grid scan pattern. We accounted for noise in the LiDAR measurements by adding unbiased Gaussian noise to each ray measurement. The standard deviation of the LiDAR noise was set such that $3\sigma = 5$ cm at an observational range of 500 meters. This standard deviation was scaled proportionally to the observational range; for instance, at 1000 meters, the standard deviation of the noise is doubled. The detector size was set to 256 x 256 pixels. The sensor field-of-view (FOV) was configured to cover a 100 x 100 meter area on the terrain when observed from a range of 500 meters with the sensor oriented directly downward. Table \ref{tab:lidar-config} shows the nominal scan coverage, GSDs, and Gaussian noise levels for each scan configuration. Note that these values are provided for reference purposes only; the simulated scans exhibit variations in coverage and GSDs due to terrain interactions and projection distortions. For scans where the FOV footprint is larger than the created testbed, any point cloud data (PCD) samples falling outside the testbed are ignored.

\subsubsection{DEM Construction Quality}

First, we evaluate the quality of DEMs reconstructed from noisy and sparse point cloud data (PCD). The analysis includes both quantitative and qualitative assessments. For the quantitative analysis, we employ two metrics: \textit{root-mean-square error (RMSE)} and \textit{negative log predictive density (NLPD)}. Let $z_i$ and $\hat{z}_i$ represent the true and estimated elevation at cell $i$ of a given DEM, respectively. RMSE and NLPD are computed as follows:
\begin{equation}
\text{RMSE} = \sqrt{\frac{1}{N} \sum_{i=1}^{N} (z_i - \hat{z}_i)^2}
\end{equation}
\begin{equation}
\text{NLPD} = \frac{1}{N} \sum_{i=1}^{N} \left( \frac{(z_i - \hat{z}_i)^2}{2\sigma_i^2} + \frac{1}{2} \log(2\pi\sigma_i^2) \right)
\end{equation}
where $\sigma_i$ is the standard deviation of the elevation at cell $i$ estimated by the model. RMSE measures the average error of the point estimate, while NLPD assesses the quality of the predicted probability distributions. This is critical because our proposed algorithm estimates both the mean and variance of the local elevation from PCD. Smaller values indicate better performance for both RMSE and NLPD.

For DEMs obtained using the baseline algorithm, NLPD is computed by setting $\sigma_i = \sigma_\text{LiDAR}$, which is the Gaussian error added to the range measurements in our LiDAR model. This approach is taken because the ALHAT algorithm, which serves as our baseline, accounts for LiDAR measurement errors when computing the probability of roughness safety.

In general, the x and y coordinates of the ground truth DEM and the DEMs estimated from the PCD by each algorithm are not the same. For analysis, we matched their coordinate by returning the value at the data point closest to the point of interpolation. 

\begin{table}[ht]
\caption{\label{tab:dem-quality-testbed} Comparison of baseline and proposed algorithms for DEM reconstruction quality across different LiDAR scan configurations.}
\centering
\begin{tabular}{cccccc}
\toprule 
\multicolumn{2}{c}{LiDAR Scan} & \multicolumn{2}{c}{RMSE, m $\downarrow$} & \multicolumn{2}{c}{NLPD $\downarrow$} \\
Range, m & Angle, deg & Baseline & Proposed & Baseline & Proposed \\
\hline \hline
200 & 0 & \bf{0.0124} & 0.0134 & \bf{-2.3387} & -1.9853 \\
200 & 30 & 0.0151 & \bf{0.0150} & -1.4848 & \bf{-2.2010} \\
200 & 60 & 0.0194 & \bf{0.0177} & 0.1489 & \bf{-2.0869} \\
500 & 0 & 0.0239 & \bf{0.0212} & -2.1439 & \bf{-2.2846} \\
500 & 30 & 0.0258 & \bf{0.0222} & -1.9777 & \bf{-2.2456} \\
500 & 60 & 0.0305 & \bf{0.0252} & -1.5031 & \bf{-2.0843} \\
1000 & 0 & 0.0380 & \bf{0.0354} & -1.8316 & \bf{-1.8350} \\
1000 & 30 & 0.0396 & \bf{0.0363} & -1.7761 & \bf{-1.8067} \\
1000 & 60 & 0.0433 & \bf{0.0409} & -1.6393 & \bf{-1.6739} \\ \bottomrule
\end{tabular}
\end{table}

Table \ref{tab:dem-quality-testbed} presents the numerical results for DEM reconstruction performance in the simulated testbed. The results are based on a simulated testbed with 500 rocks, each 1 meter in diameter, placed on a flat 200 x 200 meter surface.
Except for the optimal observation condition at the closest range (200 m) with the LiDAR directly pointing downward, the proposed algorithm outperforms the baseline algorithm. Both algorithms show a trend of decreased performance for more challenging LiDAR scan configurations, with less accuracy at longer ranges and larger observational angles.

\begin{figure}
    \centering
    \includegraphics[width=\linewidth]{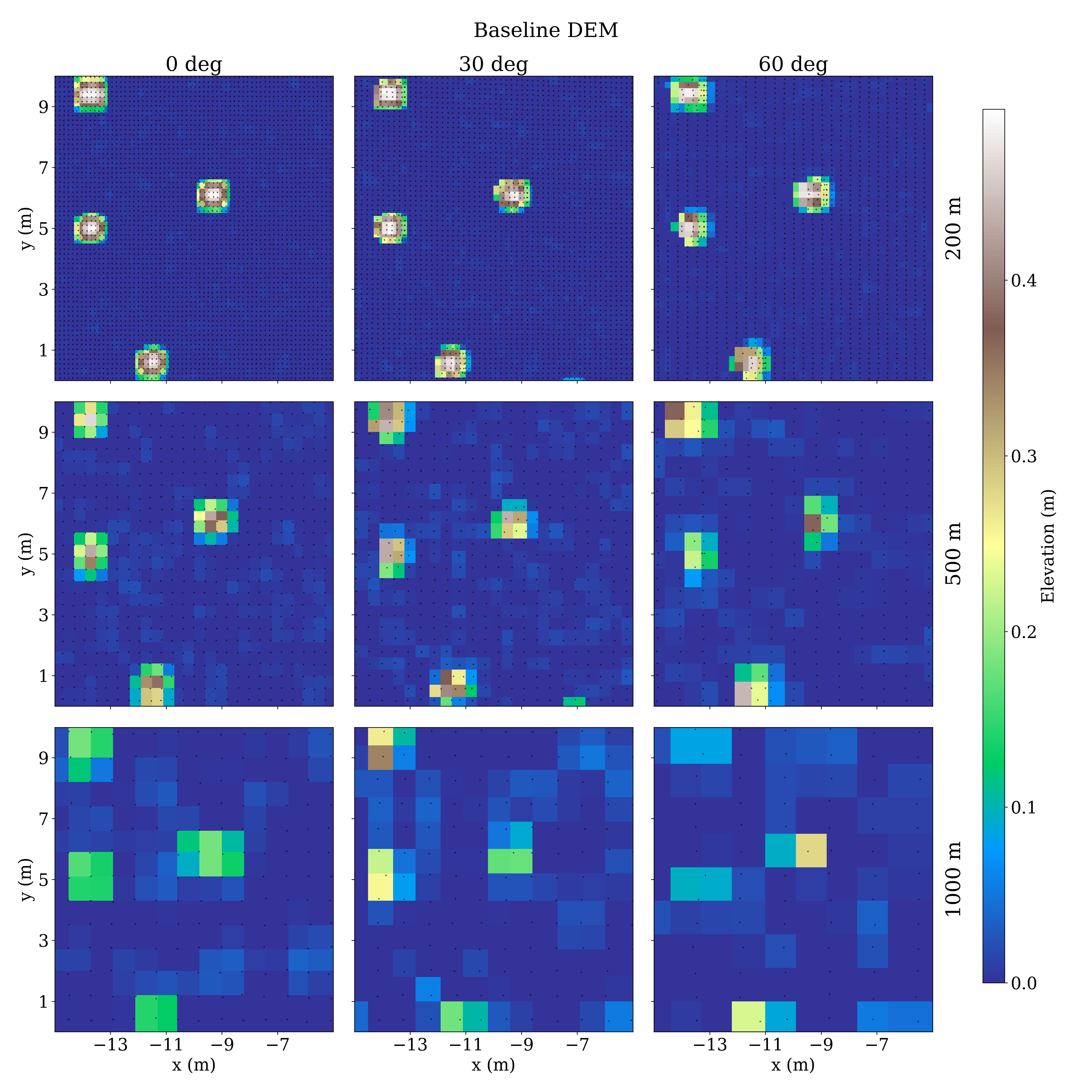}
    \caption{DEMs constructed by the baseline algorithm over a 10 x 10 meter segment for different ranges and angles of the observation. The black dots denote the LiDAR PCD.}
    \label{fig:baseline-dem}
\end{figure}

\begin{figure}
    \centering
    \includegraphics[width=\linewidth]{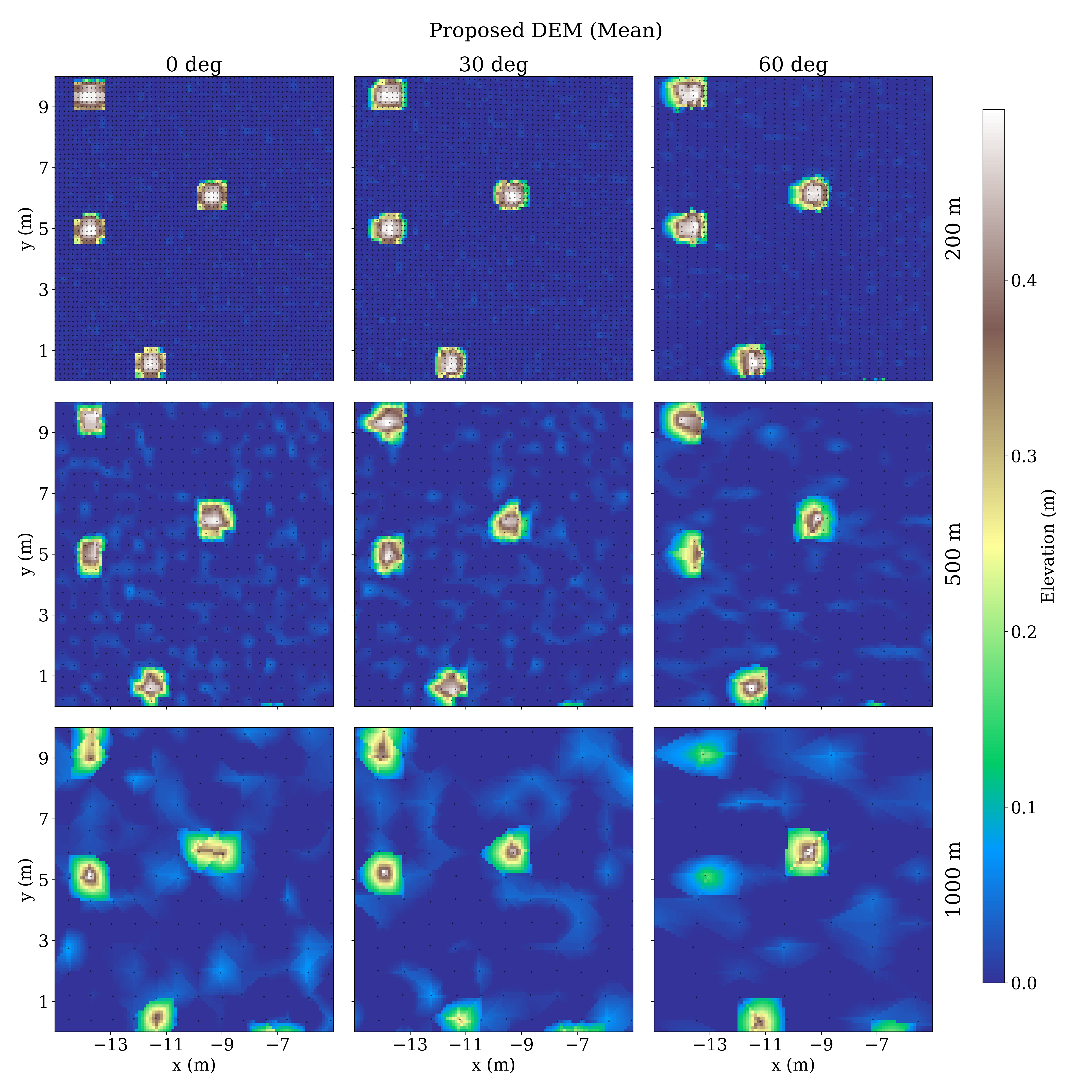}
    \caption{Mean DEMs constructed by the proposed algorithm over a 10 x 10 meter segment for different ranges and angles of the observation. The black dots denote the LiDAR PCD.}
    \label{fig:proposed-dem-mean}
\end{figure}

\begin{figure}
    \centering
    \includegraphics[width=\linewidth]{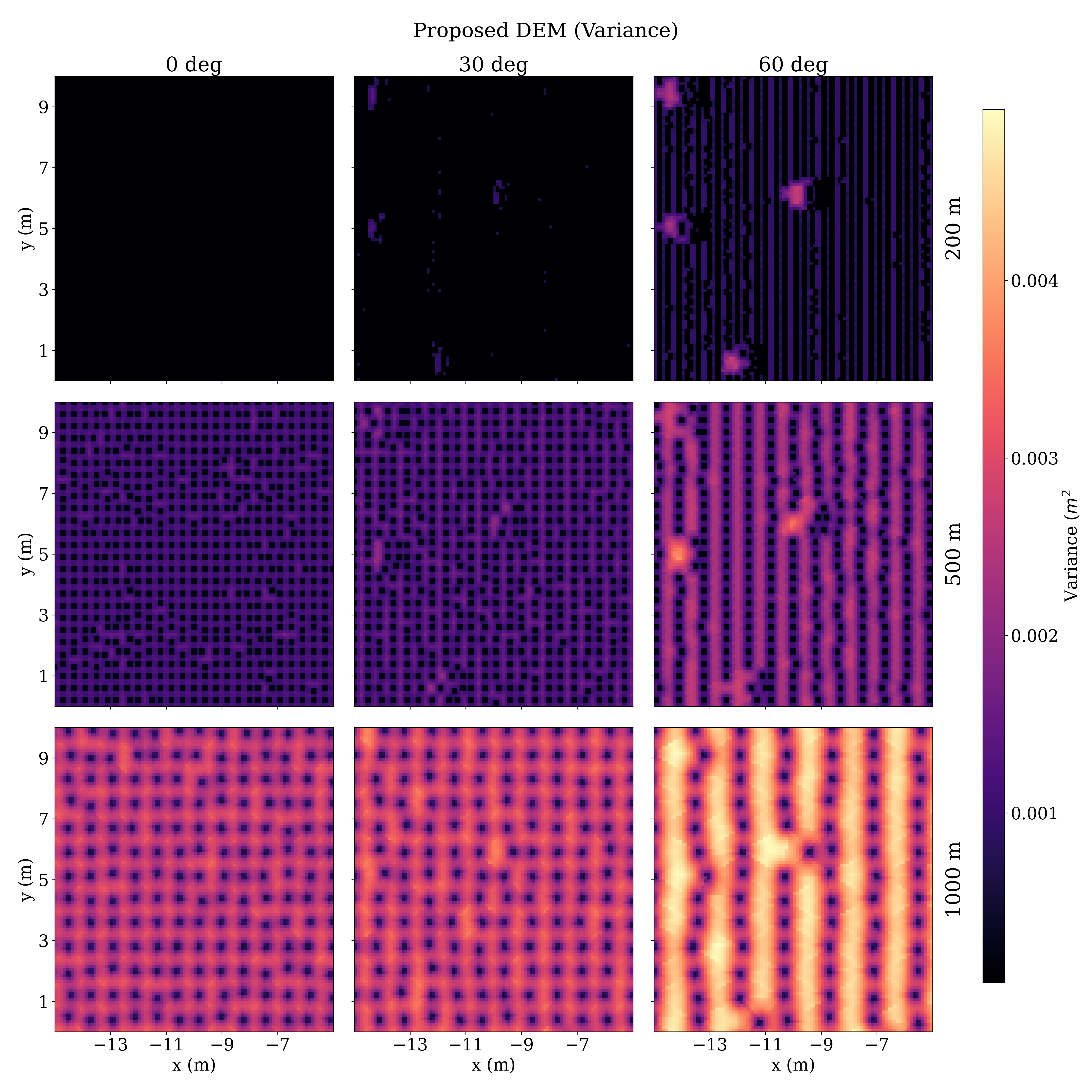}
    \caption{Variance DEMs constructed by the proposed algorithm over a 10 x 10 meter segment for different ranges and angles of the observation. The black dots denote the LiDAR PCD.}
    \label{fig:proposed-dem-var}
\end{figure}

Figure \ref{fig:baseline-dem} shows qualitative results for the baseline algorithm, while Figs. \ref{fig:proposed-dem-mean} and \ref{fig:proposed-dem-var} show the mean and variance of the estimated DEM by the proposed algorithm. Baseline DEMs suffer from severe pixelation due to enlarged GSDs of the PCD at longer observational ranges or larger observational angles. In contrast, our proposed method consistently reconstructs most of the rocks. Additionally, the DEM variance shown in Figure \ref{fig:proposed-dem-var} effectively quantifies the DEM uncertainty due to PCD sparsity. For example, when the observation is made from 60 degrees toward the positive x-direction at a distance of 200 m, occluded areas behind the rocks exhibit higher uncertainty than surrounding pixels. For larger GSDs, due to an increased observational range or angle, the DEM uncertainties increase accordingly.

\begin{figure}
    \centering
    \includegraphics[width=\linewidth]{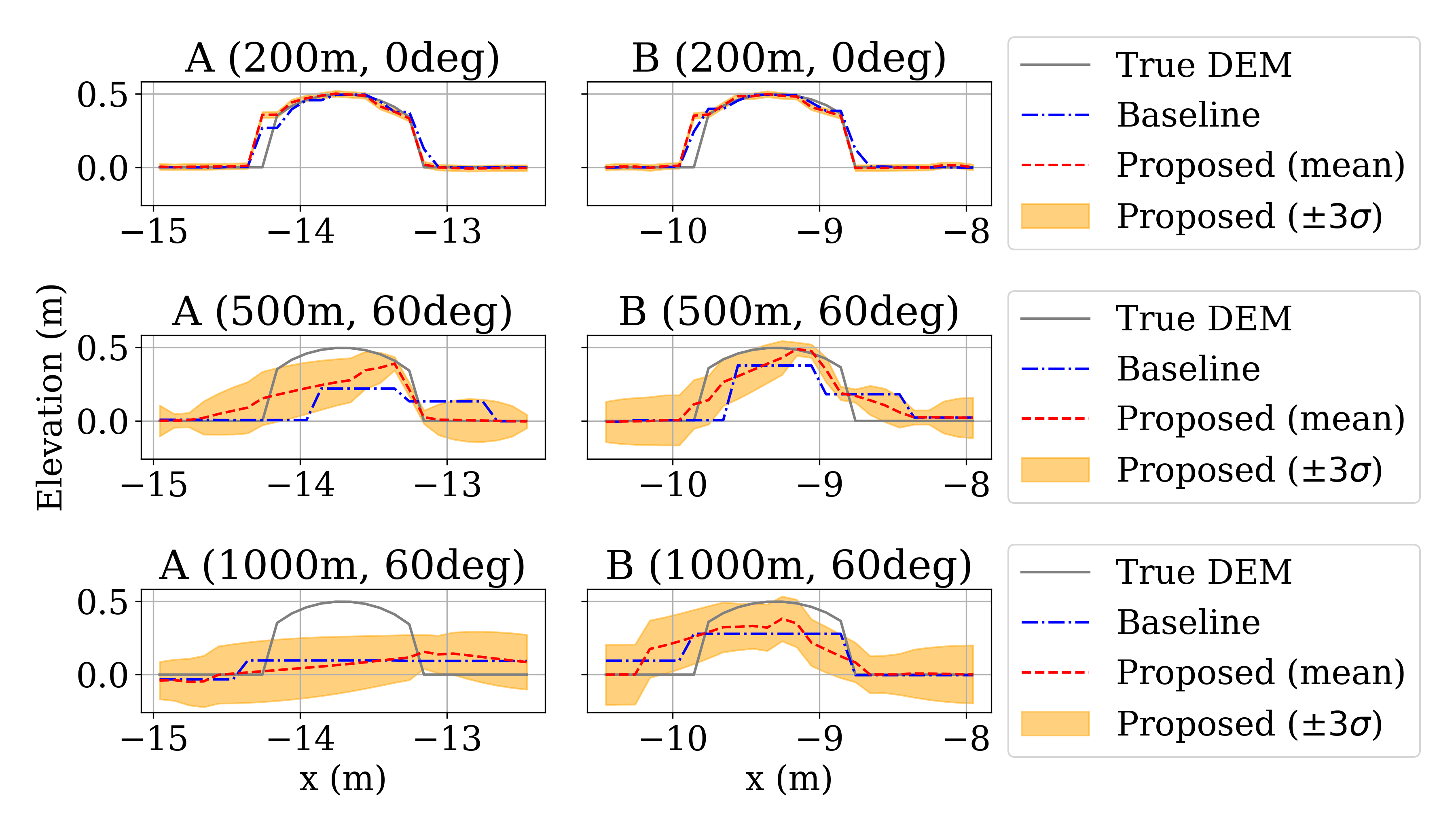}
    \caption{Comparisons of the true DEM, baseline DEM, and DEM distribution by the proposed algorithm.}
    \label{fig:rock-ab}
\end{figure}

Figure \ref{fig:rock-ab} provides a detailed analysis of the estimated DEM around rocks A and B, which are annotated in Fig. \ref{fig:rock-data-pipeline}. The left and right columns display results for rocks A and B, respectively. Each row presents results for different LiDAR scan configurations as indicated in their titles. When sufficient samples are obtained from the observation, both the baseline and proposed algorithms successfully reconstruct the rocks, as in the case of a range of 200 m and an angle of 0°. However, under more challenging observational conditions, the baseline algorithm fails to reconstruct the rocks. Since it adjusts the resolution proportional to GSD, the averaging effect reduces the estimated local elevation. In contrast, the proposed algorithm maintains a user-defined high resolution, avoiding such averaging effects. Additionally, the proposed algorithm increases local elevation uncertainty for larger GSDs, which helps represent the possibility of rock presence. However, when the GSD is too large, as in the case of a range of 1000 m and an angle of 60°, resulting in an average GSD larger than the rock diameter (see Table \ref{tab:lidar-config}), the proposed algorithm cannot always reconstruct local rocks. It can represent the rock feature for rock B but fails for rock A. Nonetheless, even for rock A, the algorithm increases the uncertainty due to the large GSD, enabling the hazard detection algorithm to detect potential landing risks, as discussed in the next section.

\subsubsection{Safety Prediction Performance}

\begin{table}[ht]
\caption{\label{tab:rock_precision_recall} Comparison of precision and recall for different LiDAR scan configurations.}
\centering
\begin{tabular}{cccccccccc} \toprule
&& \multicolumn{4}{c}{Precision $\uparrow$} & \multicolumn{4}{c}{Recall $\uparrow$}\\ \cline{3-6} \cline{7-10} 
\multicolumn{2}{c}{Scan Configuration} & \multicolumn{2}{c}{Slope} & \multicolumn{2}{c}{Roughness} & \multicolumn{2}{c}{Slope} & \multicolumn{2}{c}{Roughness} \\
Range, m & Angle, deg & Baseline & Proposed & Baseline & Proposed & Baseline & Proposed & Baseline & Proposed \\
\hline \hline
200 & 0 & 0.9999 & \bf{1.0000} & 0.9532 & \bf{1.0000} & \bf{0.9991} & 0.8226 & \bf{1.0000} & 0.9335 \\
200 & 30 & 0.9997 & \bf{1.0000} & 0.9428 & \bf{0.9991} & \bf{0.9980} & 0.8150 & \bf{0.9999} & 0.9266 \\
200 & 60 & 0.9997 & \bf{1.0000} & 0.9308 & \bf{0.9960} & \bf{0.9983} & 0.8119 & \bf{0.9999} & 0.9189 \\
500 & 0 & 0.9980 & \bf{1.0000} & 0.9245 & \bf{1.0000} & \bf{1.0000} & 0.8214 & \bf{1.0000} & 0.9313 \\
500 & 30 & 0.9980 & \bf{1.0000} & 0.9664 & \bf{0.9987} & \bf{1.0000} & 0.8257 & \bf{0.9949} & 0.9269 \\
500 & 60 & 0.9967 & \bf{1.0000} & 0.8986 & \bf{0.9982} & \bf{1.0000} & 0.8450 & \bf{0.9996} & 0.9238 \\
1000 & 0 & 0.9888 & \bf{0.9985} & 0.7970 & \bf{0.9973} & \bf{1.0000} & 0.8966 & \bf{1.0000} & 0.9242 \\
1000 & 30 & 0.9969 & \bf{0.9991} & 0.8352 & \bf{0.9967} & \bf{1.0000} & 0.8980 & \bf{0.9995} & 0.9144 \\
1000 & 60 & 0.9911 & \bf{0.9987} & 0.7859 & \bf{0.9573} & \bf{1.0000} & 0.9318 & \bf{0.9998} & 0.8828 \\ \bottomrule
\end{tabular}
\end{table}

Next, we evaluate the safety prediction performance using both quantitative and qualitative methods. For quantitative analysis, we use precision (Eq. \eqref{eq: precision}) and recall (Eq. \eqref{eq: recall}) for both slope and roughness hazards. To compute the precision and recall, any stochastic estimate of landing safety is converted to a binary estimate using a threshold of 0.5; for instance, if $P(\text{Safe}) > 0.5$, then we label that pixel as safe.

\begin{figure}
    \centering
    \includegraphics[width=\linewidth]{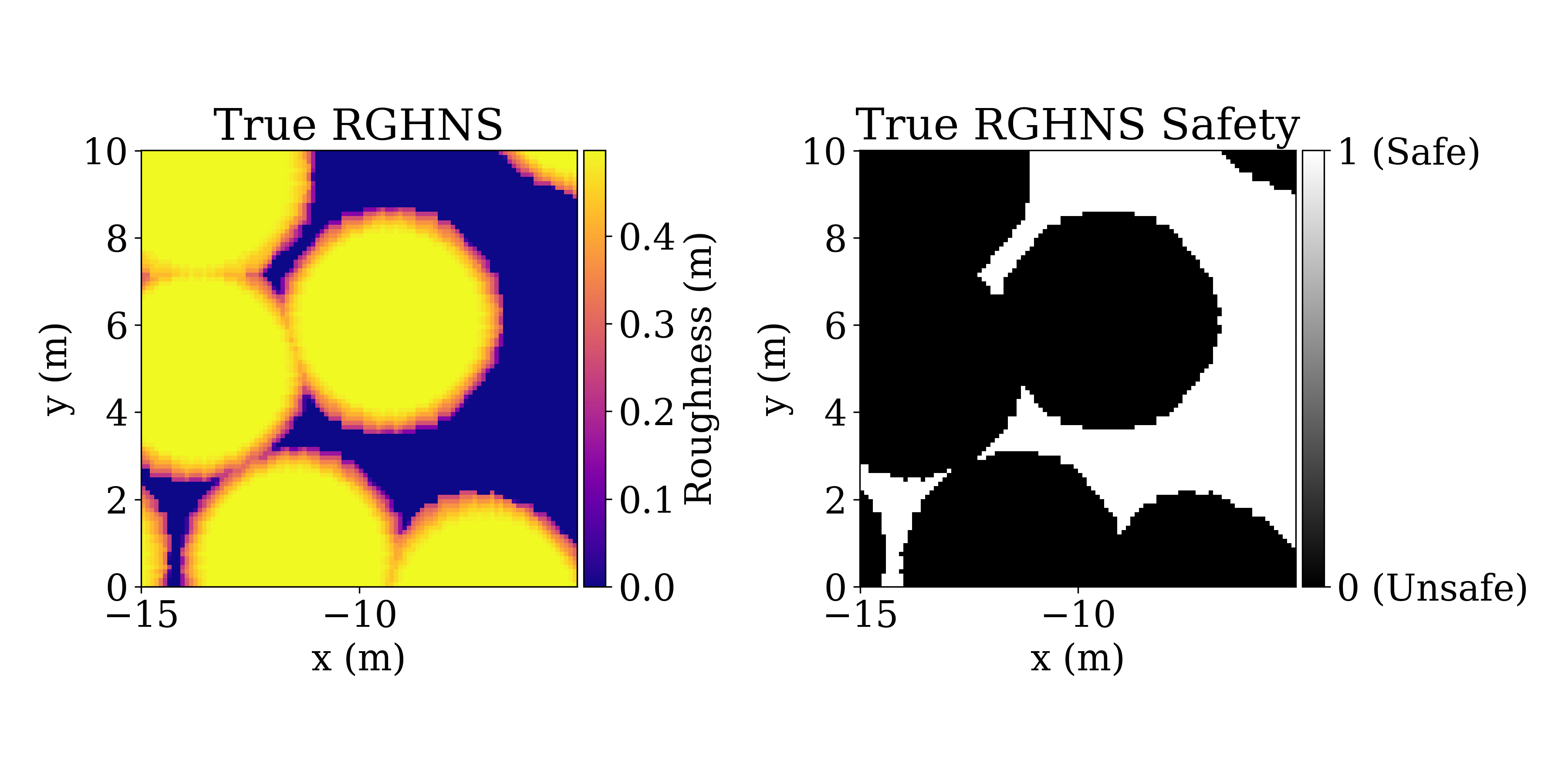}
    \caption{Ground truth roughness and roughness-based safety over a 10 x 10 meter segment.}
    \label{fig:rghns-gt}
\end{figure}

\begin{figure}
    \centering
    \includegraphics[width=\linewidth]{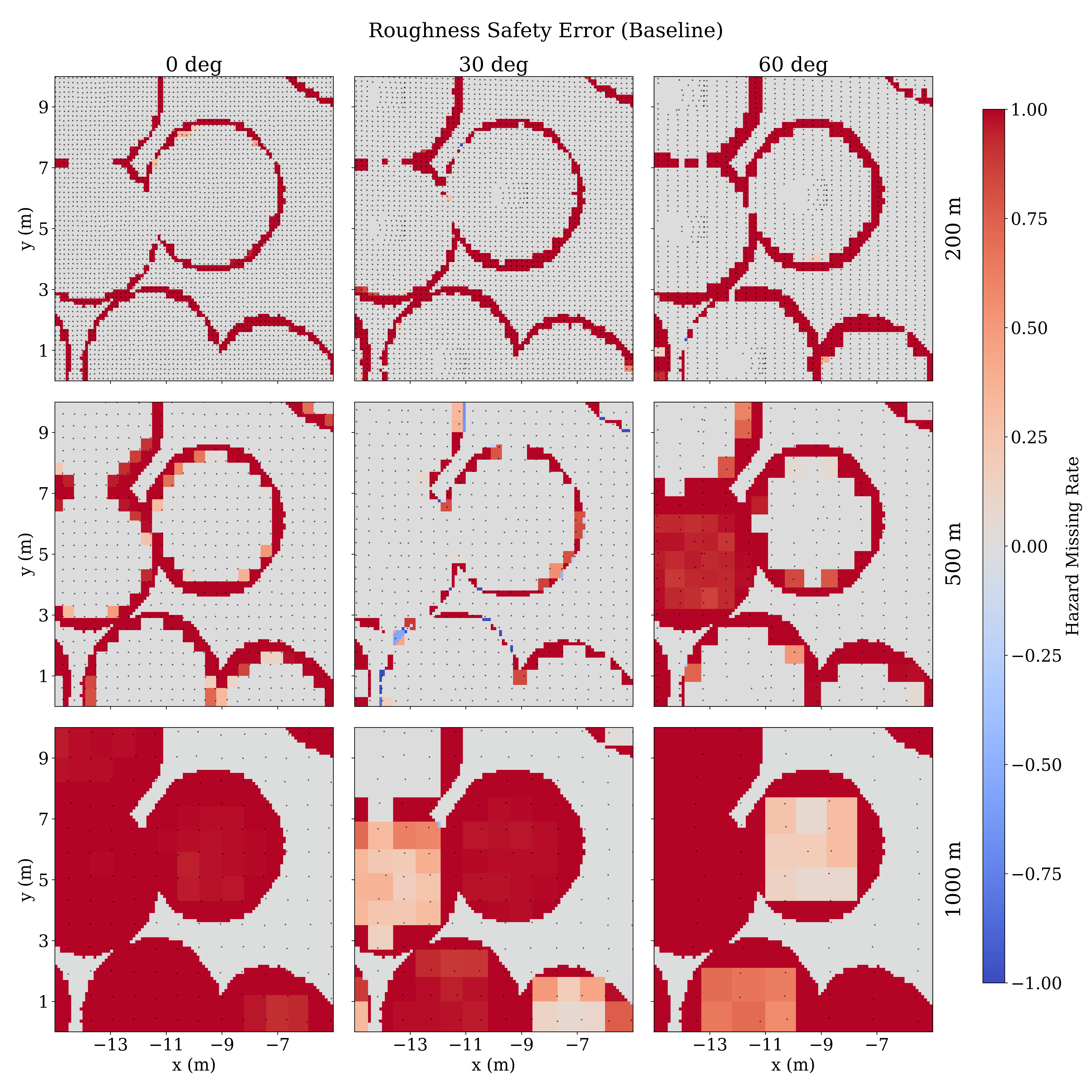}
    \caption{Roughness hazard missing rates over a 10 x 10 meter segment for the baseline algorithm. The black dots denote the LiDAR PCD.}
    \label{fig:baseline-rghns}
\end{figure}

\begin{figure}
    \centering
    \includegraphics[width=\linewidth]{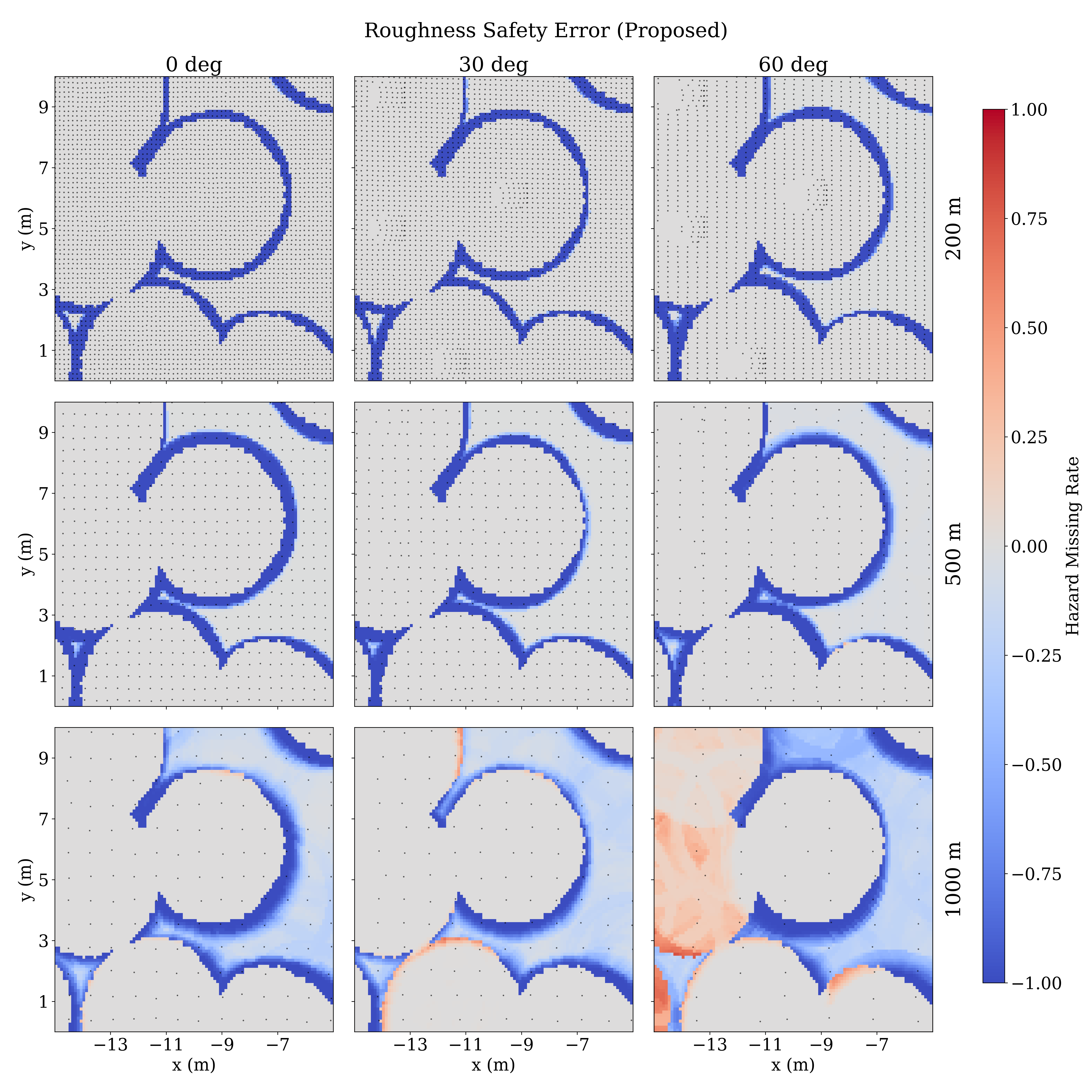}
    \caption{Roughness hazard missing rates over a 10 x 10 meter segment for the proposed algorithm. The black dots denote the LiDAR PCD.}
    \label{fig:proposed-rghns}
\end{figure}

Figure \ref{fig:rghns-gt} shows the true roughness and roughness safety over a 10 x 10 meter segment. To evaluate the safety map estimation performance, we visualized the hazard missing rate in Figs. \ref{fig:baseline-rghns} and \ref{fig:proposed-rghns} for the baseline and proposed algorithms, respectively. The hazard missing rate is defined as $(\text{True Safety}) - (\text{Estimated Safety Probability}) \in [-1, 1]$. A hazard missing rate of 0 indicates an accurate prediction, closer to 1 indicates more missed hazards, and closer to -1 indicates overly conservative hazard assessments. For the baseline algorithm, challenging observational conditions with larger observational ranges and angles result in more missed roughness hazards due to the averaging effect caused by coarse resolution, as detailed in the previous section.

Note that the baseline algorithm results in some missing hazards at the hazard boundaries even under moderate observation conditions (e.g., observations from 200 m). This is also due to the adjusted pixel sizes. The ground truth DEMs and safety maps are generated with a 0.1 m/pix resolution, and DEMs with larger pixel sizes result in missing hazards around the boundary due to rounding errors.

Figure \ref{fig:proposed-rghns} demonstrates that the proposed algorithm can detect roughness hazards even under challenging observational conditions. For observations at a range of 1000 m, the algorithm increases safety uncertainty in response to the large GSDs, resulting in a reduced hazard missing rate. Quantitative results for all cases are summarized in Table~\ref{tab:rock_precision_recall}.

\begin{figure}
    \centering
    \includegraphics[width=\linewidth]{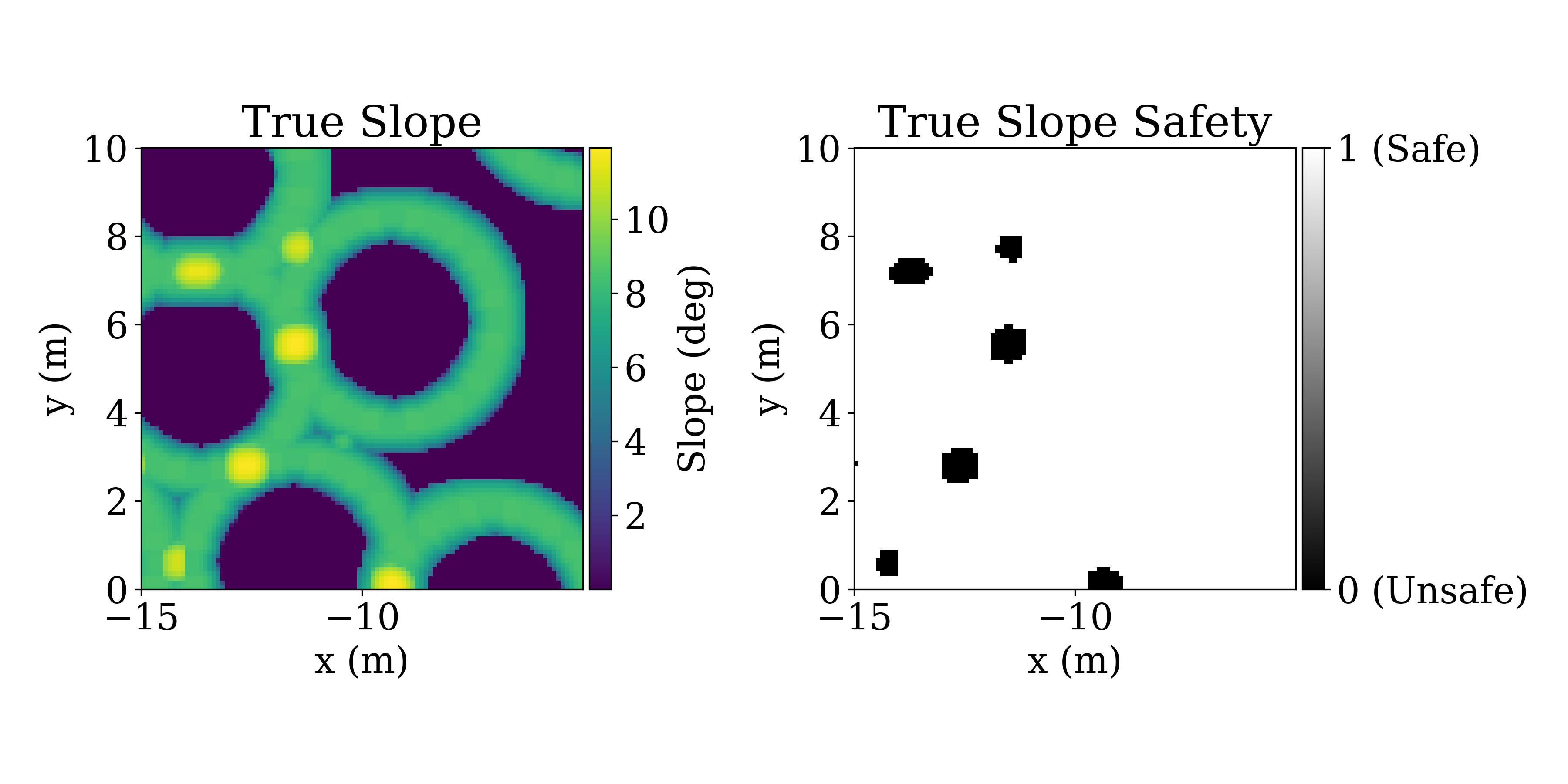}
    \caption{Ground truth slope and slope-based safety over a 10 x 10 meter segment.}
    \label{fig:slope-gt}
\end{figure}

\begin{figure}
    \centering
    \includegraphics[width=\linewidth]{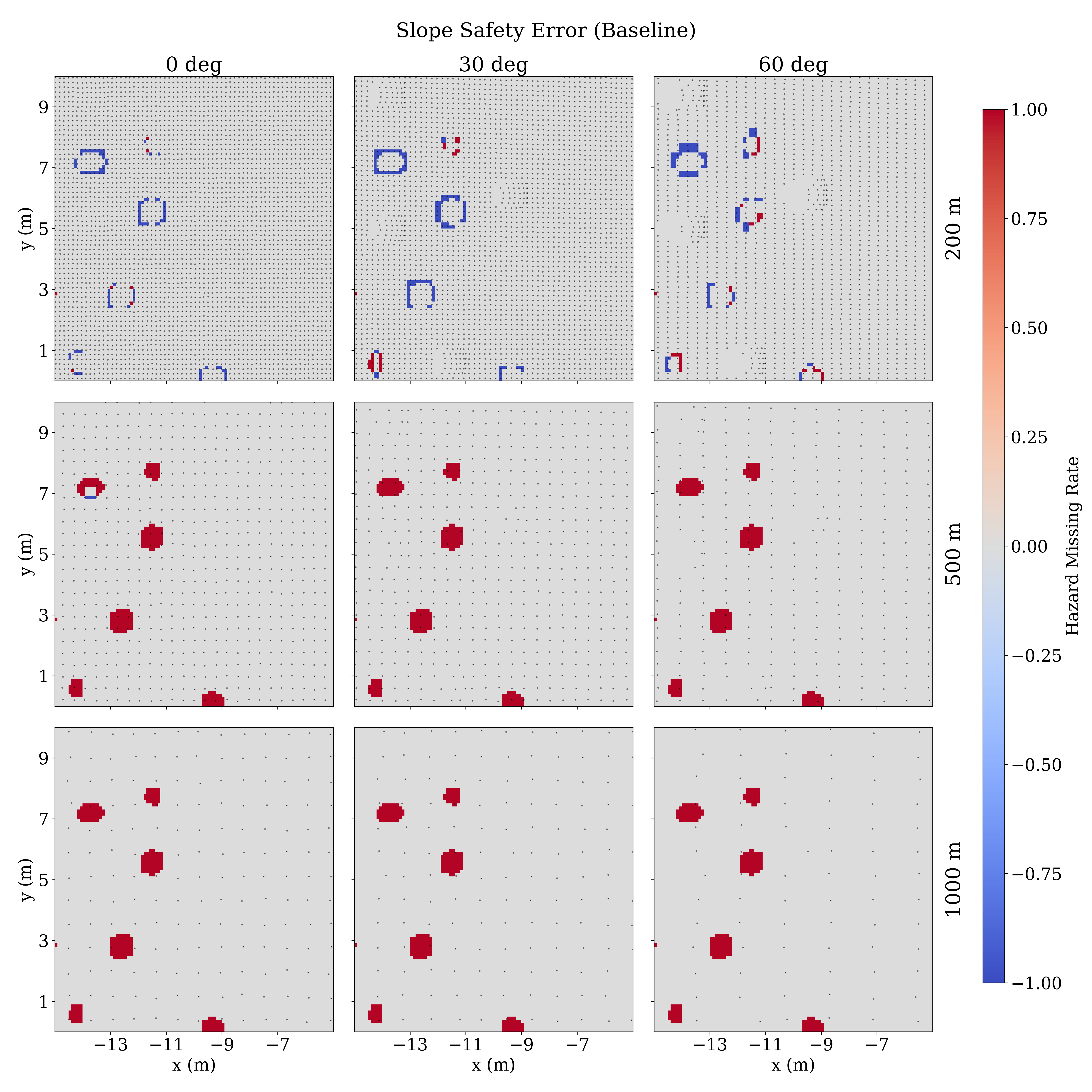}
    \caption{Slope hazard missing rates over a 10 x 10 meter segment for the baseline algorithm. The black dots denote the LiDAR PCD.}
    \label{fig:baseline-slope}
\end{figure}

\begin{figure}
    \centering
    \includegraphics[width=\linewidth]{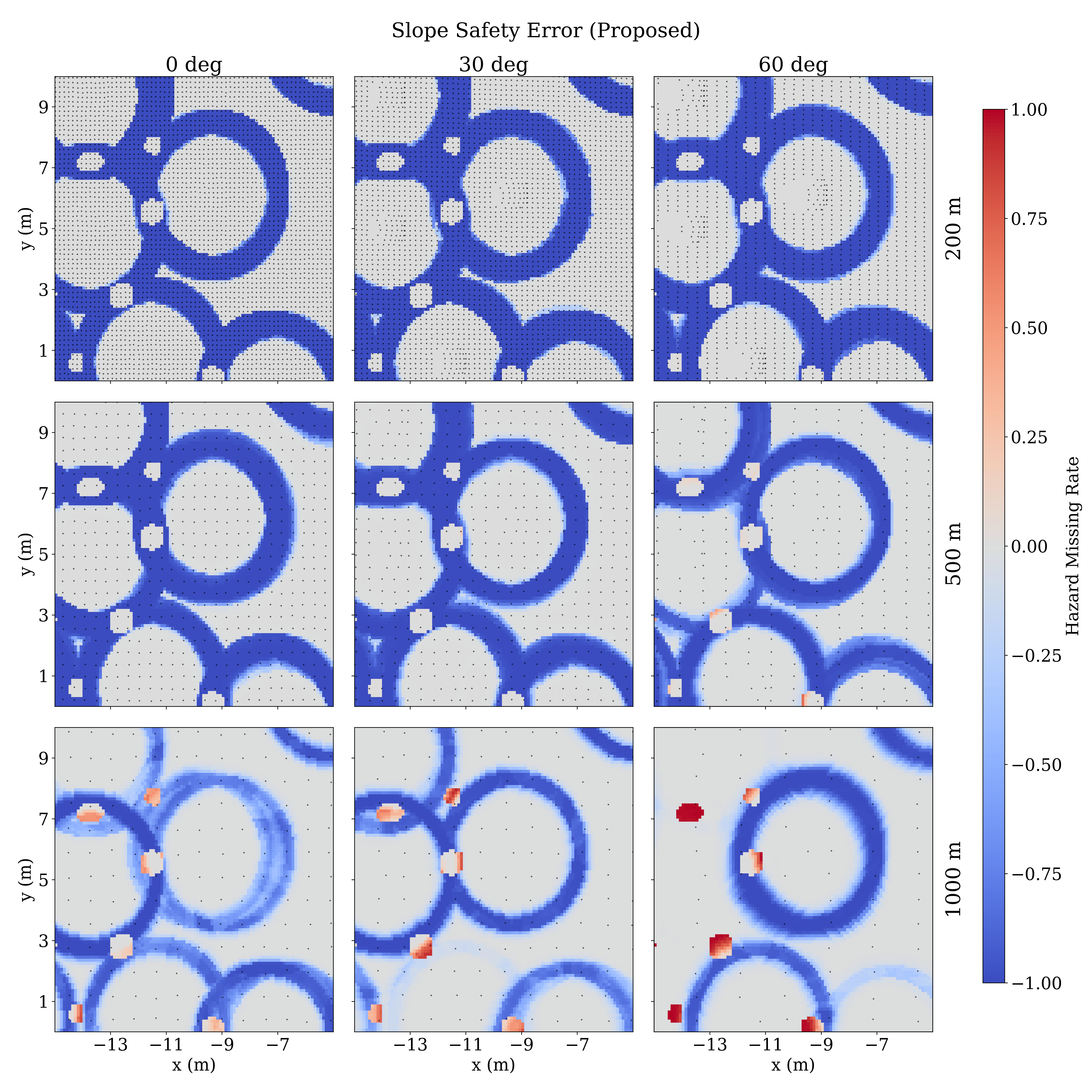}
    \caption{Slope hazard missing rates over a 10 x 10 meter segment for the proposed algorithm. The black dots denote the LiDAR PCD.}
    \label{fig:proposed-slope}
\end{figure}

Figure \ref{fig:slope-gt} shows the true slope and slope hazard map. Figures \ref{fig:baseline-slope} and \ref{fig:proposed-slope} demonstrate that the proposed algorithm has a better hazard detection rate compared to the baseline, similar to the roughness analysis. The main difference is that the baseline algorithm misses hazards under more moderate observational conditions, while the proposed algorithm provides a more conservative evaluation of hazards than in the roughness case. This occurs because, in this testbed scenario, the slope values are at the boundary between safe and unsafe, as indicated in Figure \ref{fig:slope-gt}, making the trend of missed hazards and conservative evaluations by both algorithms more pronounced.

\subsection{Performance Analysis with Realistic Data}
To evaluate the algorithm's performance on realistic data, we selected ten distinct locations on the Mars Digital Terrain Model (DTM) for the candidate ExoMars landing site in Hypanis Valles~\cite{hirise}. The visualization of this realistic dataset with region specifications is given in Appendix~\ref{appdx: real-data}.

Here, we also conducted a sensitivity analysis on the hyperparameter $\ell$ of the Gaussian Random Field (GRF) regression described in Eq. \eqref{eq:ae-kernel-realtimeshd}. The prior standard deviation of elevation $\sigma_f$ was computed based on the distribution of the Point Cloud Data (PCD) elevations. Specifically, Table \ref{tab:real_precision_recall} presents the comparison of precision and recall for various $\ell$ values at a range of 500 m with the terrain sensor oriented nadir. The general trend indicates that the proposed algorithm is more conservative than the baseline, which aligns with the results from the simulated testbeds; the proposed algorithm demonstrates higher precision, while the baseline exhibits higher recall. Table \ref{tab:real_precision_recall} further demonstrates that smaller values of $\ell$ result in more conservative estimates, yielding higher precision and lower recall. This behavior is consistent with the physical interpretation of $\ell$; the hyperparameter $\ell$ defines the degree of correlation between nearby elevations, as reflected in the kernel's formulation in Eq. \eqref{eq:ae-kernel-realtimeshd}.

This trend is visually demonstrated by Figs. \ref{fig:slope-real} and \ref{fig:rghns-real}, which illustrate the results for region R4 with $\ell=1.0$ m. In particular, Fig. \ref{fig:rghns-real} highlights that the baseline algorithm fails to detect most roughness hazards, whereas the proposed algorithm successfully identifies the majority of them. 

\begin{table}[ht]
    \caption{\label{tab:real_precision_recall} Comparison of precision and recall for different $\ell$ values at 500m range; $\ell$ in meters.}
    \centering
    \begin{tabular}{ccccccccc} \toprule
     & \multicolumn{4}{c}{Slope} & \multicolumn{4}{c}{Roughness} \\
    \cline{2-5} \cline{6-9}
    &Baseline& \multicolumn{3}{c}{Proposed} &Baseline& \multicolumn{3}{c}{Proposed} \\
    \cline{3-5} \cline{7-9} 
    &  & $\ell=0.3$ & $\ell=1.0$ & $\ell=5.0$ &  & $\ell=0.3$ & $\ell=1.0$ & $\ell=5.0$ \\
    \hline \hline
    Precision $\uparrow$ & 0.9515 & \bf{0.9999} & \bf{0.9999} & 0.9997 & 0.8741 & \bf{0.9995} & 0.9906 & 0.9795 \\
    Recall $\uparrow$ & \bf{0.9960} & 0.3417 & 0.5113 & 0.6098 & \bf{0.9999} & 0.0030 & 0.0562 & 0.1389 \\
    \bottomrule
    \end{tabular}
\end{table}

\begin{figure}
    \centering
    \includegraphics[width=0.75\linewidth]{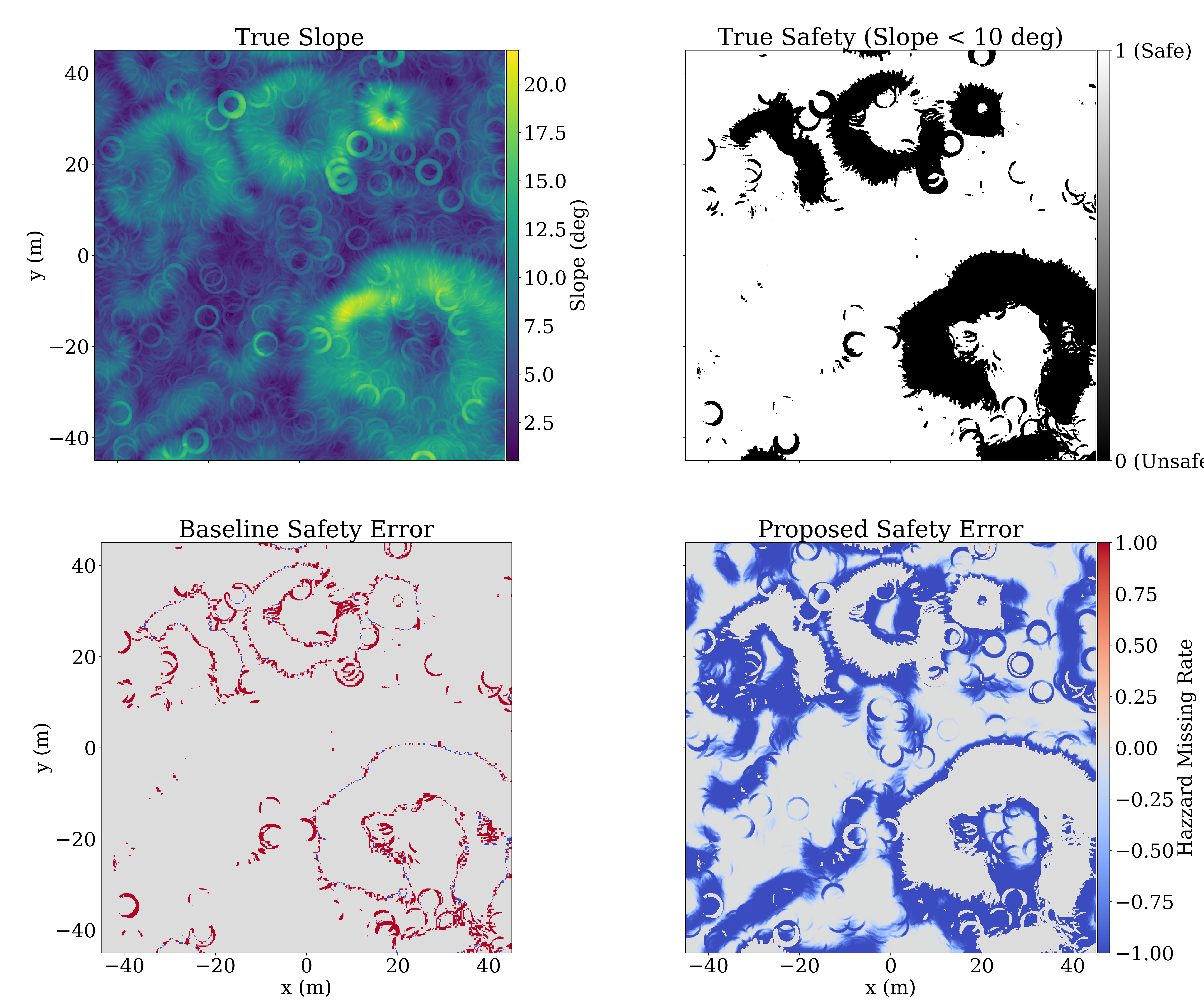}
    \caption{Performance comparison of baseline and proposed algorithms for realistic terrain: slope case.}
    \label{fig:slope-real}
\end{figure}

\begin{figure}
    \centering
    \includegraphics[width=0.75\linewidth]{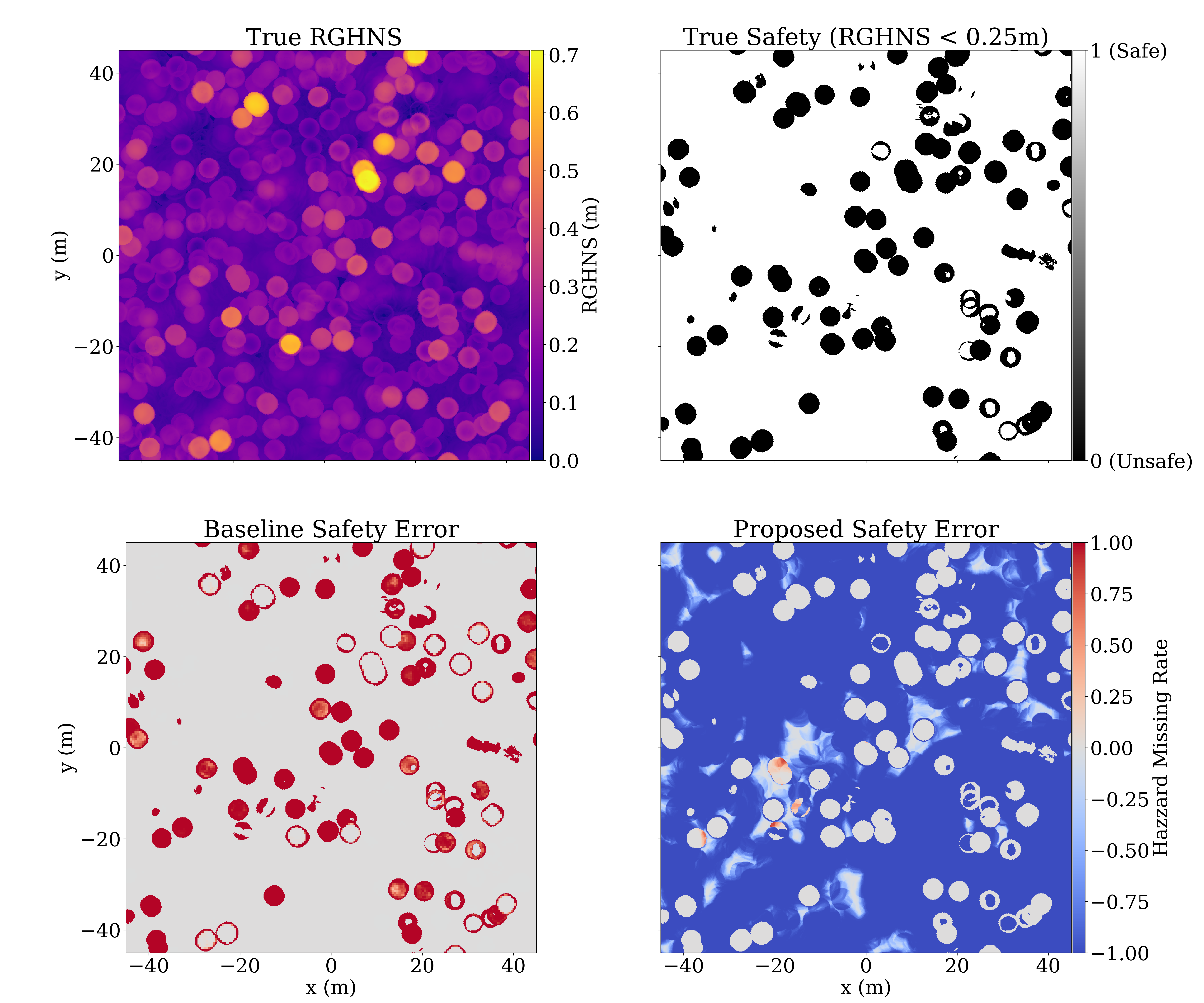}
    \caption{Performance comparison of baseline and proposed algorithms for realistic terrain: roughness case.}
    \label{fig:rghns-real}
\end{figure}

\subsection{Computational Cost Analysis}

\begin{table}[htbp!]
\centering
\caption{\label{tab:runtime} Run time comparison across different map resolutions.}
\begin{threeparttable}
\begin{tabular}{ccccc}
\toprule
Resolution,  & \multicolumn{2}{c}{{Baseline}, s $\downarrow$} & \multicolumn{2}{c}{{Proposed}, s $\downarrow$} \\
m/pix & DEM$^*$ & Safety & DEM & Safety \\
\hline \hline
0.5 & - & 0.0412 & 0.1248 & 0.0016 \\
0.3 & - & 0.3890 & 0.2991 & 0.0469 \\
0.2 & - & 2.5194 & 0.7112 & 0.0724 \\
0.1 & - & 84.3238 & 2.7666 & 1.0420 \\
\bottomrule
\end{tabular}
\begin{tablenotes}
\item {\footnotesize $^*$ The baseline DEM generation algorithm cannot specify the resolution.}
\end{tablenotes}
\end{threeparttable}
\end{table}

Our contributions include not only improved hazard detection but also reduced computational cost. Table \ref{tab:runtime} shows the run time comparison between the proposed and baseline algorithms. We measured run time for a 100x100 m terrain observed at an altitude of 500 m with LiDAR directly pointing downward, whose detector size is 256x256. Our proposed algorithm can generate the DEM and safety maps at user-defined resolutions, and we measured the run time across a set of resolutions: 0.1, 0.2, 0.3, and 0.5 m/pix. Note that we measured the computational time on a standard desktop PC, and the purpose of this analysis is to compare the run time between the proposed and baseline algorithms for different resolutions; the actual run time depends on the choice of hardware for each mission. 

The terrain mapping and processing for HDA consists of two steps: PCD to DEM, and DEM to safety map. We proposed new algorithms for both parts, and the run times for each part are given in Table \ref{tab:runtime}. Note that the ALHAT algorithm, which serves as the baseline algorithm in this paper, adjusts the DEM resolution with the GSD of the PCD, so we only measured the DEM to safety map part for the baseline.

Even ignoring the DEM construction run time of the baseline algorithm, the proposed algorithm achieves faster computation except for the coarse resolution of 0.5 m/pix. In higher resolution cases, the sum of the run times of PCD to DEM and DEM to safety map by the proposed algorithm is smaller than the run time of DEM to safety map by the baseline algorithm. Comparing DEM to safety map run times between the resolution of 0.2 and 0.1 m/pix cases, we can confirm that the baseline algorithm takes $O(p^5)$ while the proposed algorithm takes $O(p^4)$, which aligns with the algorithm's design (Appendix \ref{appdx: scalability}). Here, $p$ is the resolution, which is the number of pixels that span a unit length.

\section{Conclusion}
Following the presentation of the detailed background and literature review, this paper addressed the challenge of landing safety evaluation with topographic uncertainty caused by the sparse measurement of terrain sensors, a critical issue for achieving reliable landing safety evaluations. We presented a novel modeling technique for topographic uncertainty using Gaussian random fields (GRF), effectively capturing the variability and uncertainty in terrain data resulting from sparse sensor measurements. 

We proposed Gaussian Digital Elevation Maps, which is a stochastic variant of conventional Digital Elevation Maps (DEMs) and each cell contains means and variances of the local elevation. We designed a novel real-time Gaussian DEM construction algorithm by applying Delauney triangulation and local GRF regression. 

We also developed a novel real-time stochastic landing safety evaluation algorithm. First, we designed a provably conservative real-time hazard detection algorithm for regular DEM input. The geometric investigation of the lander-terrain interaction is exploited to efficiently evaluate the conservative local slope and roughness only based on height differences while avoiding the costly computation of the landing plane. Then, the developed algorithm was extended for Gaussian DEM input. 

The effectiveness of these methods was demonstrated through detailed studies with simulated testbeds. The results showed that the proposed algorithm evaluated landing safety more reliably than the state-of-the-art (SOTA) algorithm, particularly in scenarios with sparse terrain measurements; at maximum, the proposed algorithm increased the precision by 20\%. As for computational performance, for high-resolution input of 10 meters per pixel for 100 x 100 m terrain, the proposed algorithm improved computation time more than 20 times; reduced to 3.8 seconds from 84.3 seconds by the SOTA algorithm.
This highlighted the practical applicability and robustness of the proposed methods in realistic space mission conditions.

\section*{Appendix}

\subsection{Scalability of Hazard Detection Algorithms}\label{appdx: scalability}
For fixed sizes of the digital elevation map (DEM) length and the lander diameter, their pixel span scales proportionally to $p$, the unit length of the pixel size. To evaluate the roughness safety, we need to compute all the pixels within the lander footprint, which scales $O(p^2)$. We repeat this process for all possible lander orientations and for every DEM pixel. The number of lander orientation steps is driven by the width of the lander in pixels. Therefore, the exact hazard detection (HD) algorithm scales as $O(p^5)$~\cite{johnson2022OPTIMIZATIONLIDARBASED}. The proposed provably conservative HD algorithm allows us to conservatively skip repeated evaluation for all possible orientations, resulting in scalability $O(p^4)$.

\subsection{Proof of Theorem \ref{theorem:max-slope-tri}}\label{appdx: proof}

For a given triangle formed by $\bm{p}_1$, $\bm{p}_2$, and $\bm{p}_3$, without loss of generality, we can take a local-vertical local-horizontal (LVLH) frame $\mathcal{F}^0=\{\bm{i}^0, \bm{j}^0, \bm{k}^0\}$ such that $\bm{p}_1=\bm{0}$,  $\bm{p}_2=[p_{2, x}, 0, p_{2, z}]$, $\bm{p}_3=[p_{3, x}, p_{3, z}, p_{3, z}]$, and all the nonzero elements are positive except for $p_{3, x}$.
As shown in Fig. \ref{fig:triangle-geom}, the triangle $\bm{p}_1\bm{p}_2\bm{p}_3$ is reached by applying a sequence of intrinsic rotations to a triangle on the $\bm{i}\bm{j}$ plane. 
Specifically, let $R_0^1$ be the rotation about $\bm{j}^0$ by $\theta$, a rotation from $\mathcal{F}^0$ to $\mathcal{F}^1=\{\bm{i}^1, \bm{j}^1, \bm{k}^1\}$. Similarly, let $R_1^2$ be the rotation about $\bm{i}^1$ by $\phi$, a rotation from $\mathcal{F}^1$ to $\mathcal{F}^2=\{\bm{i}^2, \bm{j}^2, \bm{k}^2\}$. 
In $\mathcal{F}^2$, the triangle is located on the $\bm{i}^2\bm{j}^2$ plane. Let the vertex coordinates in $\mathcal{F}^2$ be $\bm{p}^2_1=\bm{0}$,  $\bm{p}^2_2=[x_2, 0, 0]$, and $\bm{p}^2_3=[x_3, y_3, 0]$ where all non-zero elements are positive except for $x_3$. 
Then, the vertex locations in $\mathcal{F}^0$ can be computed as $P^0 = R_0^1 R_1^2 P^2$ where $P^0=[\bm{p}_1$, $\bm{p}_2$, $\bm{p}_3]\in\R^{3\times 3}$ and $P^2=[\bm{p}^2_1$, $\bm{p}^2_2$, $\bm{p}^2_3]\in\R^{3\times 3}$ are the matrices representing the triangle vertex locations in $\mathcal{F}_0$ and $\mathcal{F}_2$, respectively. The element-wise results are as follows. 
\begin{equation}\label{eq: p2p3}
    \bm{p}_2=
    \begin{bmatrix}
    x_2 \cos\theta\\ 0 \\ x_2\sin\theta
    \end{bmatrix}, \quad
    \bm{p}_3=
    \begin{bmatrix}
    x_3\cos\theta - y_3\sin\theta\sin\phi \\ y_3\cos\phi \\ x_3\sin\theta + y_3\cos\theta\sin\phi
    \end{bmatrix}
\end{equation}
The triangle's normal vector facing $z>0$ halfspace is computed by $(\bm{p}_2 - \bm{p}_1) \times (\bm{p}_3 - \bm{p}_1)$. Then, using $x_2, y_3>0$, the triangle's slope is obtained as follows.
\begin{equation}\label{eq: triangle-slope}
        s = \cos^{-1}\left(\frac{x_2 y_3 \cos\theta \cos\phi}{\sqrt{x_2^2 y_3^2}}\right)=\cos^{-1}\left(\cos\theta \cos\phi\right).
\end{equation}

\begin{figure}
    \centering
    \includegraphics[width=0.75\linewidth]{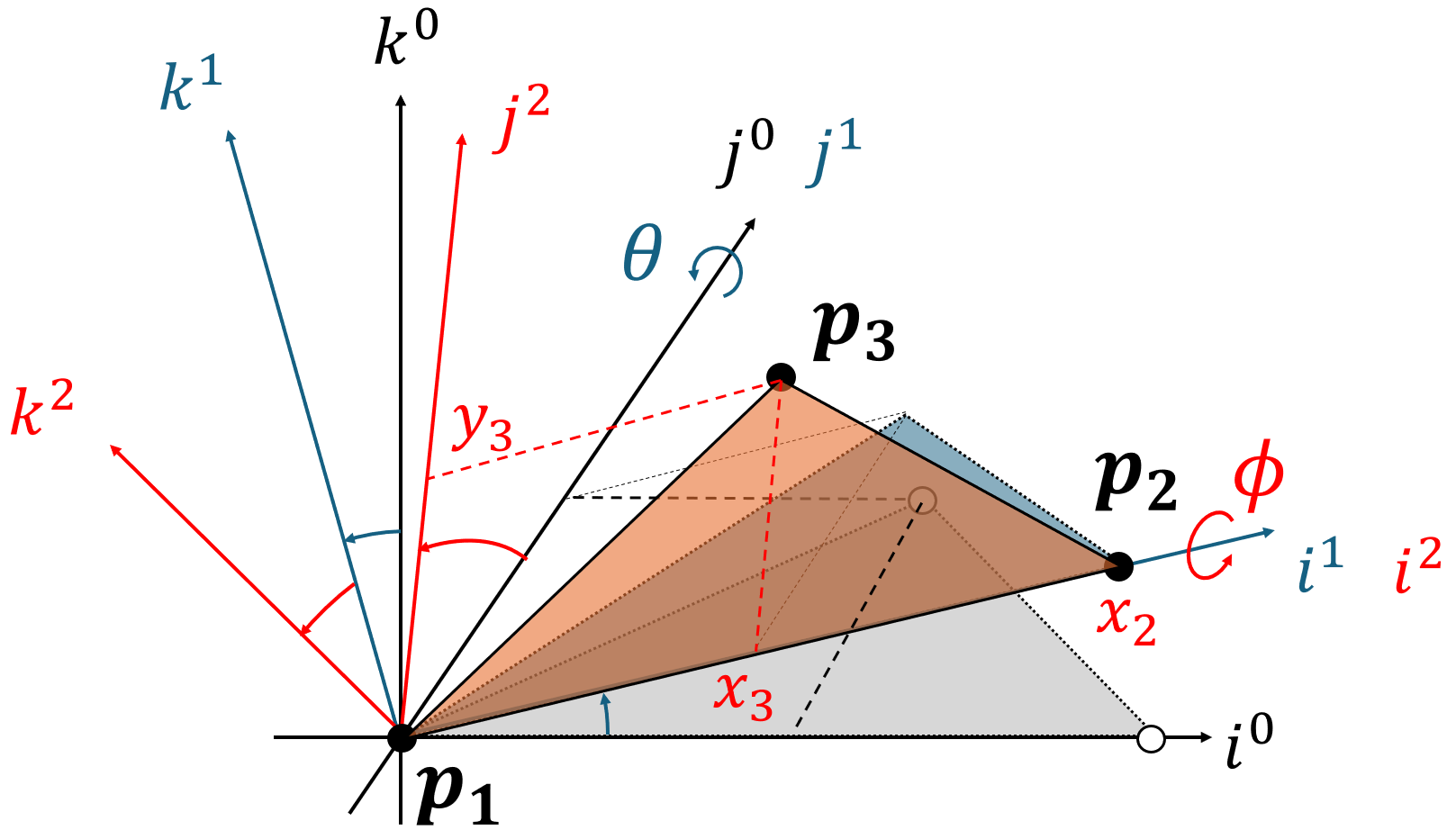}
    \caption{A triangle orientation represented by $\theta$ and $\phi$.}
    \label{fig:triangle-geom}
\end{figure}

The assumption by Theorem 1 bounds the range of $\theta$ and $\phi$. Substituting Eq. \eqref{eq: p2p3} into the bounded elevations of vertices, i.e., $0\leq p_{2, z}, p_{3, z} \leq \Delta z$, we have the following bounds on $\theta$ and $\phi$.
\begin{equation}\label{eq: angle-bounds}
    0\leq \sin\theta \leq \frac{\Delta z}{x_2}, \quad -\frac{x_3}{y_3}\tan\theta \leq \sin\phi \leq \frac{\Delta z - x_3\sin\theta}{y_3\cos\theta}
\end{equation}

The distance between the vertices $1$ and $2$ is larger than or equal to the distance from the vertex $1$ to the opposite edge, i.e., $x_2 \geq h_1$. Due to the assumption that $\Delta z < h_i$ for all $i=1, 2, 3$, we obtain $\Delta z /x_2 <1$. Therefore, the lower and the upper bounds of Eq. \eqref{eq: angle-bounds} for $\theta$ are both reachable at $p_{2, z}=0$ and $p_{2, z}=\Delta z$, respectively. For later use, let $\bar{\theta}$ denote the upper bound; $\bar{\theta}= \sin^{-1}(\Delta z / x_2)$. The range of $\theta$ becomes $[0, \bar{\theta}]$. 

The distance from vertex $2$ to the opposite edge, $h_2$, is obtained as $h_2=|x_2y_3|/\sqrt{x_3^2 + y_3^2}$. Substituting this to $\Delta z < h_2$ and reformulating the inequality, we obtain the following.
\begin{equation}
    \frac{y_3}{|x_3|} > \frac{\frac{\Delta z}{x_2}}{\sqrt{1- \left(\frac{\Delta z}{x_2}\right)^2}} \geq \tan\theta
\end{equation}
The inequality on the right was obtained by using $1 > \Delta z/x_2 \geq \sin\theta$ with $\tan\theta = \cos\theta / \sqrt{1-\sin^2\theta}$. Therefore, if $x_3\geq 0$, the lower bound for $\sin\phi$ in Eq. \eqref{eq: angle-bounds} is larger than $-1$. Otherwise, it is larger than $0$. In either case, the lower bound of $\phi$ in Eq. \eqref{eq: angle-bounds} is reachable at $p_{3, z}=0$.

To show the upper bound of $\sin\phi$ is reachable, define $f(\theta)$ as the right-hand side of Eq. \eqref{eq: angle-bounds} and compute the derivative as follows.
\begin{equation}
    f(\theta)\triangleq \frac{\Delta z - x_3\sin\theta}{y_3\cos\theta}, \quad
    f'(\theta)=\frac{\Delta z \sin\theta - x_3}{y_3\cos^2\theta}
\end{equation}
If $x_3>0$ and $x_3/\Delta z < \Delta z/x_2 $, the sign of $f'(\theta)$ shifts from negative to positive as $\theta$ grows. If $x_3\leq 0$, then $f'(\theta)$ remain positive. After all, it is sufficient to show $f(0) < 1$ and $f(\bar{\theta})<1$ to show that $f(\theta)<1$ for $\theta\in[0, \bar{\theta}]$ and that the upper bound of $\sin\phi$ is reachable. $f(0)=\Delta z / y_3 <1$ is immediate because $h_3=y_3$ and $\Delta z < h_3$. $f(\bar{\theta})<1$ can be shown using $\sin\bar{\theta}=\Delta z / x_2$ and $\Delta z > h_1=x_2y_3/\sqrt{(x_2-x_3)^2 + y_3^2}$, where $h_1$ is the distance from vertex 1 to the opposite edge. 

Since the lower and the upper bounds for $\sin\phi$ of Eq. \eqref{eq: angle-bounds} are both reachable, $\sin\phi$ is at one of the bounds when the triangle's slope of Eq. \eqref{eq: triangle-slope} is maximized. When $\sin\phi$ is at the lower bound, the triangle's slope becomes as follows.
\begin{equation}
    s(\theta) = \cos^{-1}\left(\cos\theta\sqrt{1-\frac{x_3^2}{y_3^2}\tan^2\theta}\right)
\end{equation}
which is maximized at $\theta=\bar{\theta}$. This maximum slope is achieved at $(p_{z, 1}, p_{z, 2}, p_{z, 2})=(0, \Delta z, 0)$.

Similarly, when $\sin\phi$ is at the upper bound, the triangle's slope is obtained as follows.
\begin{equation}
\begin{split}
    s(\theta) &= \cos^{-1}\left(\cos\theta\sqrt{1-\frac{(\Delta z - x_3\sin\theta)^2}{y_3^2\cos^2\theta}}\right)\\
    &=\cos^{-1}\left(\frac{1}{y_3}\sqrt{ - (x_3^2 + y_3^2)\sin^2\theta + 2x_3\Delta z \sin\theta + (y_3^2 - \Delta z^2)
    }\right)
\end{split}
\end{equation}
which is maximized at either $\theta=0$ or $\theta=\bar{\theta}$ because the argument inside the square root is a concave quadratic function of $\sin\theta$. The maximum slope in this case is achieved at either $(p_{z, 1}, p_{z, 2}, p_{z, 2})=(0, 0, \Delta z)$ or $(0, \Delta z, \Delta z)$.

When $(p_{z, 1}, p_{z, 2}, p_{z, 2})=(0, \Delta z, 0)$, $(0, 0, \Delta z)$ or $(0, \Delta z, \Delta z)$, the resulting slope is computed as $s=\sin^{-1}(\Delta z / h_2)$, $s=\sin^{-1}(\Delta z / h_3)$, or $s=\sin^{-1}(\Delta z / h_1)$, respectively.
\hspace*{\fill}$\square$

\subsection{Realistic Terrain Data}\label{appdx: real-data}
To evaluate the algorithm's performance on realistic data, we selected ten distinct locations on the Mars Digital Terrain Model (DTM) for the candidate ExoMars landing site in Hypanis Valles~\cite{hirise}, as illustrated in Fig. \ref{fig:real-hirise}. Randomly placed rocks were superimposed on these terrains to simulate realistic surface conditions. Figures \ref{fig:real-dems} and \ref{fig:real-safeties} present the corresponding DEMs and safety maps for the selected regions.

\begin{figure}
\centering
\includegraphics[width=0.5\linewidth]{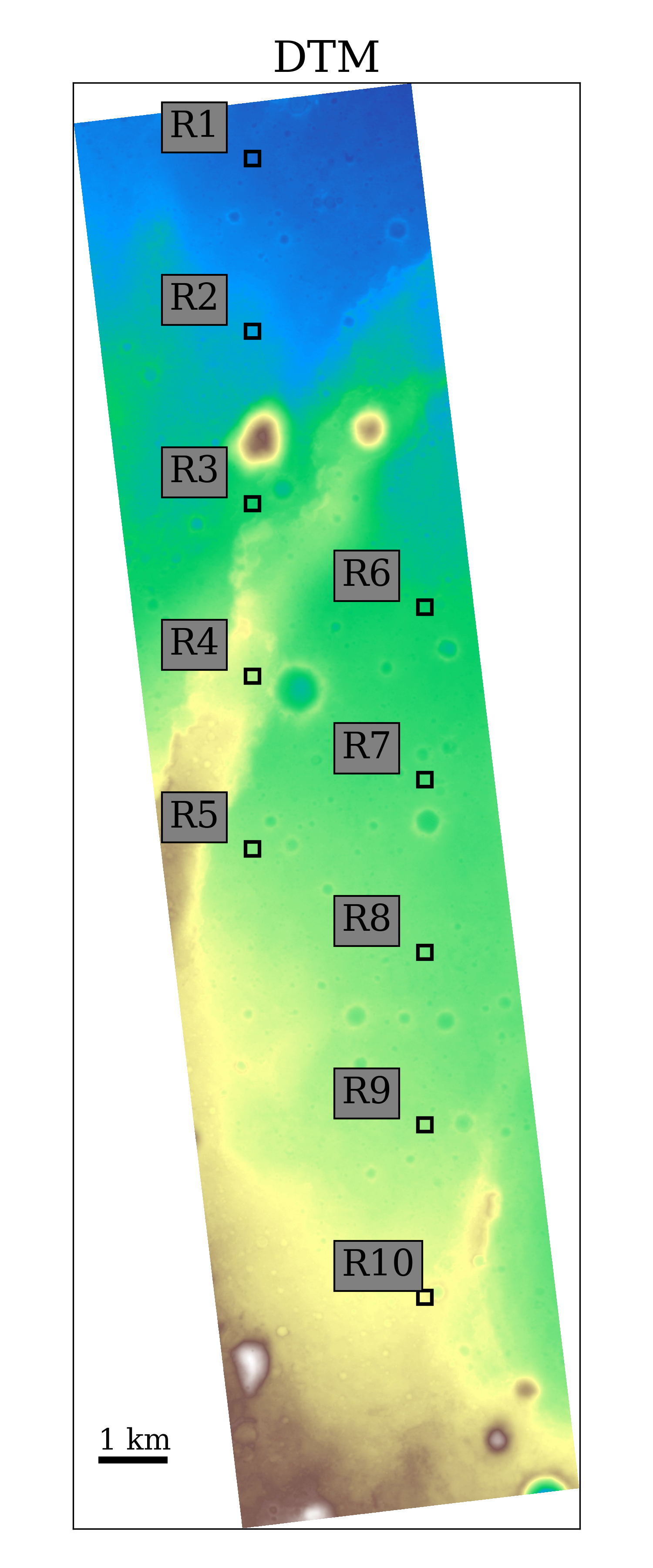}
\caption{Ten regions selected for high-fidelity analysis.}
\label{fig:real-hirise}
\end{figure}

\begin{figure}
\centering
\includegraphics[width=\linewidth]{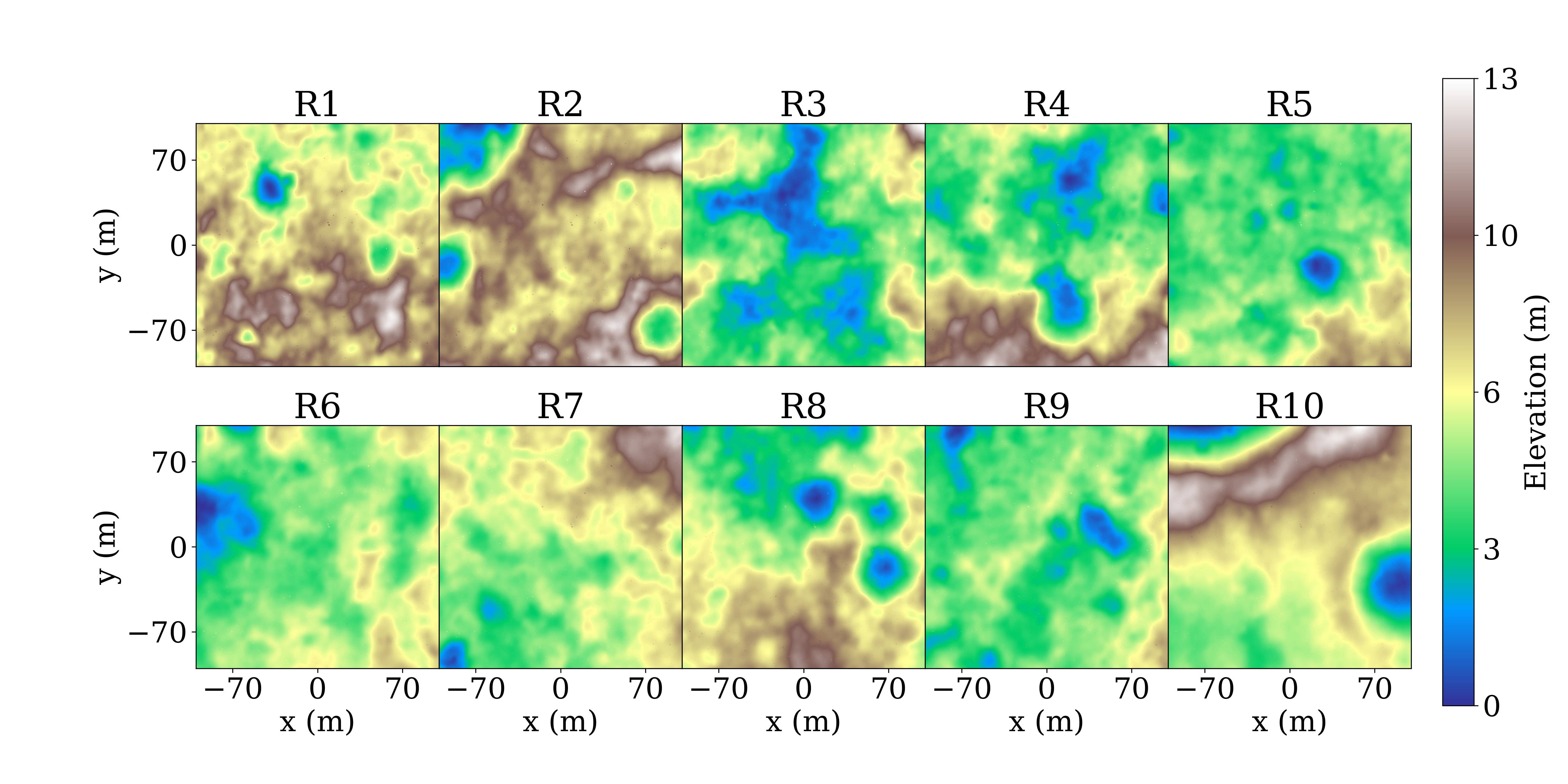}
\caption{Digital Elevation Models (DEMs) of the selected regions.}
\label{fig:real-dems}
\end{figure}

\begin{figure}
\centering
\includegraphics[width=\linewidth]{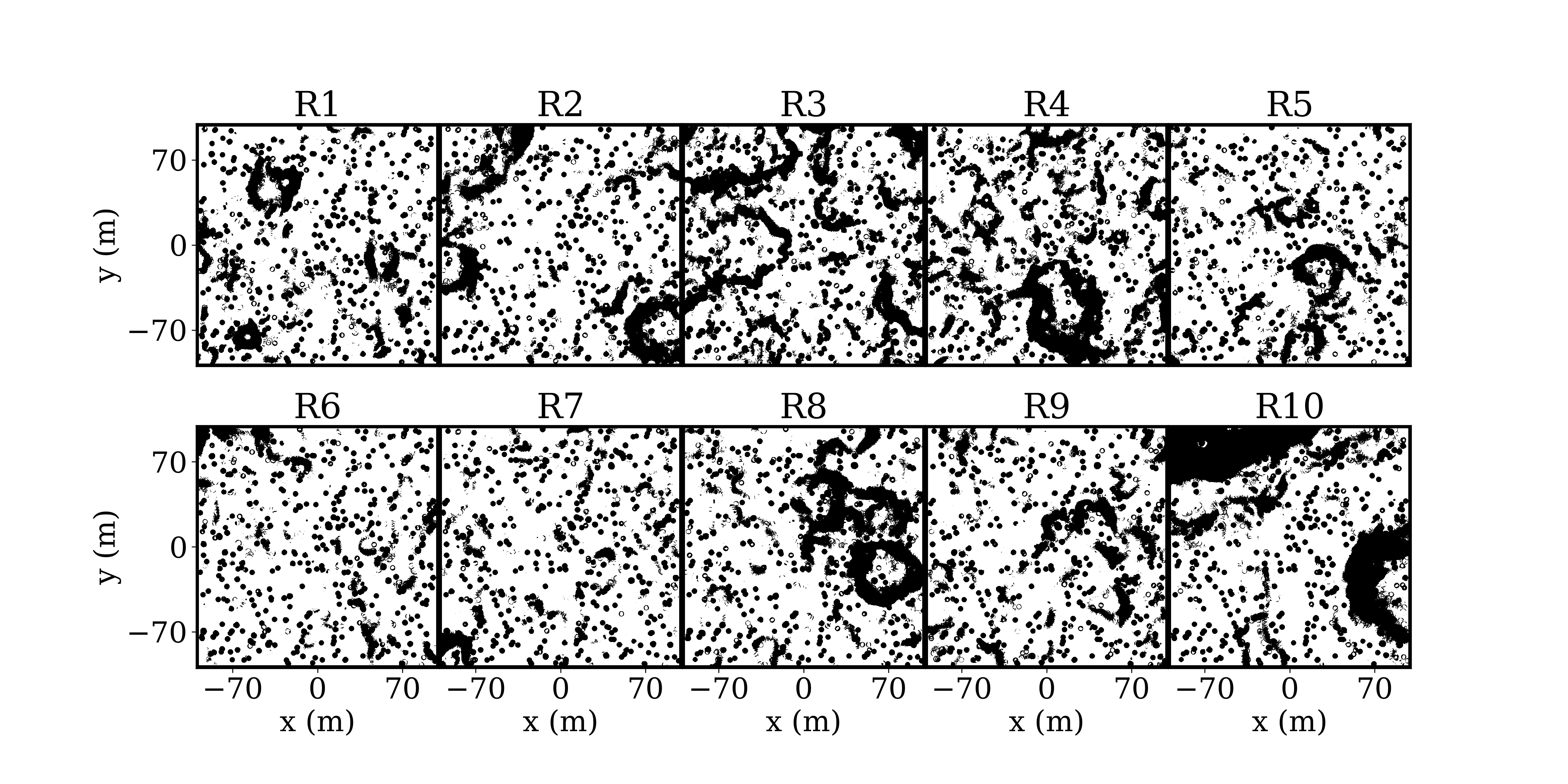}
\caption{Safety maps corresponding to the DEMs of the selected regions where white and black indicate safe and unsafe locations, respectively.}
\label{fig:real-safeties}
\end{figure}

\bibliography{references}

\end{document}